\documentclass[onecolumn,nofootinbib,thightenlines,notitlepage,tightenlines,longbibliography,superscriptaddress,11pt]{revtex4-1} 
\newcommand{\pagenumbaa}{1}
\usepackage{graphicx}
\usepackage{amssymb}
\usepackage{amstext}
\usepackage[caption=false]{subfig}
\usepackage{enumerate}
\usepackage[usenames,dvipsnames]{xcolor}
\usepackage{tikz}
\usetikzlibrary{chains}
\usetikzlibrary{fit}
\usepackage{pgflibraryarrows}		
\usepackage{pgflibrarysnakes}		

\usepackage{epsfig}
\usetikzlibrary{shapes.symbols,patterns} 
\usepackage{pgfplots}

\usepackage{amsthm}
\usepackage{amsmath}
\usepackage{amsfonts}
\usepackage{amssymb,amstext}
\usepackage{bbm} 
\usepackage{dsfont}

\usepackage[colorlinks=true,urlcolor=blue, hyperindex,breaklinks=true] {hyperref}
\usepackage{todonotes}

\usepackage{footnote}

\theoremstyle{plain}
\newtheorem{mythm}{Theorem}[section]
\newtheorem{myprop}[mythm]{Proposition}
\newtheorem{mycor}[mythm]{Corollary}
\newtheorem{mylem}[mythm]{Lemma}

\theoremstyle{definition}
\newtheorem{mydef}[mythm]{Definition}

\newcommand{\herm}{\ensuremath{^{\dagger}}}

\newcommand{\ket}[1]{\left| #1 \right \rangle}
\newcommand{\ketbra}[2]{\left| #1 \right \rangle\!\left \langle #2 \right|}

\newcommand{\norm}[1]{\left\lVert#1\right\rVert}

\newcommand{\braket}[2]{\left| #1 \right \rangle \! \left \langle #1 \right|}

\newcommand{\Hp}{\mathrm{H}_{+}^{n}} 
\newcommand{\Hop}{\mathrm{H}^{n}}

\def\Hb{h}

\def\setC{\mathsf{c}}

\newcommand{\introsection}[1]{{\bf{#1}.}}
\newcommand{\Hh}[1]{H\!\left({#1}\right)} 
\newcommand{\Hc}[2]{H\!\left({#1}\!\left|{#2}\right.\right)} 
\newcommand{\I}[2]{I\!\left({#1}:{#2}\right)} 
\newcommand{\Icc}[2]{I_{{\mathsf{c}}}\!\left({#1},{#2}\right)} 
\newcommand{\Ic}[3]{I\!\left({#1}:{#2}\!\left|{#3} \right. \right)} 

\newcommand{\1}{{\openone}}

\def\trnorm{\mathrm{tr}}
\def\opnorm{\mathrm{op}}

\newcommand{\cII}{\mathcal{I}} 

\newcommand{\Tr}[1]{{\rm tr}\!\left[{#1}\right]} 
\newcommand{\Trp}[2]{{\rm tr}_{{#1}}\!\left({#2}\right)} 

\DeclareMathOperator{\st}{s.t.}
\DeclareMathOperator{\id}{id}

\newcommand{\transp}{\ensuremath{^{\scriptscriptstyle{\top}}}}

\newcommand{\R}{\ensuremath{\mathbb{R}}}
\newcommand{\Rp}{\ensuremath{\R_{\geq 0}}}

\newcommand{\brackett}[1]{ \left \lbrace \text{#1}\right \rbrace }

\allowdisplaybreaks

\usepackage{makecell}

\begin{document}

\title{Approximate Degradable Quantum Channels}

 \author{David Sutter}
 \email[]{$\brackett{suttedav,\,scholz,\,renner}$@phys.ethz.ch}
 \affiliation{Institute for Theoretical Physics, ETH Zurich, Switzerland}
 
 \author{Volkher B.~Scholz}
 \email[]{$\brackett{suttedav,\,scholz,\,renner}$@phys.ethz.ch}
 \affiliation{Institute for Theoretical Physics, ETH Zurich, Switzerland}
 
  \author{Andreas Winter}
 \email[]{andreas.winter@uab.cat}
 \affiliation{ICREA \& F\'isica Te\`orica: Informaci\'o i Fen\`omens Qu\`antics,\\
 Universitat Aut\`onoma de Barcelona, ES-08193 Bellaterra (Barcelona), Spain}
 
  \author{Renato Renner}
 \email[]{$\brackett{suttedav,\,scholz,\,renner}$@phys.ethz.ch}
 \affiliation{Institute for Theoretical Physics, ETH Zurich, Switzerland}


\begin{abstract}
Degradable quantum channels are an important class of completely positive trace-preserving maps. Among other properties, they offer a single-letter formula  for the quantum and the private classical capacity and are characterized by the fact that a complementary channel can be obtained from the channel by applying a degrading channel. In this work we introduce the concept of approximate degradable channels, which satisfy this condition up to some finite $\varepsilon\geq0$. That is, there exists a degrading channel which upon composition with the channel is $\varepsilon$-close in the diamond norm to the complementary channel. We show that for any fixed channel the smallest such $\varepsilon$ can be efficiently determined via a semidefinite program. Moreover, these approximate degradable channels also approximately inherit all other properties of degradable channels. As an application, we derive improved upper bounds to the quantum and private classical capacity for certain channels of interest in quantum communication.
\end{abstract}

 \maketitle

 \setcounter{page}{\pagenumbaa}  
 \thispagestyle{plain}
\section{Introduction} \label{sec:intro}
The highest rate at which quantum information can be transmitted asymptotically reliably per channel use is called \emph{quantum capacity}. The \emph{private classical capacity} of a quantum channel characterizes the highest possible rate at which classical information can be transmitted asymptotically reliably per channel use such that no information about the message leaks to the environment.  
Both of these quantities are mathematically characterized by a multi-letter expression, using regularization, that is complicated to evaluate --- as a matter of fact, it is not even known to be computable~\cite{cubitt15,elkouss15}. In general, it is even difficult to derive good upper and lower bounds that can be evaluated efficiently for the two capacities.

For \emph{degradable channels}, which are characterized by the feature that the complementary channel can be written as a composition of the main channel with a degrading channel, the channel's coherent and private classical information are additive and coincide. As a result, the regularized expressions describing the quantum and private classical capacity reduce to the same single-letter formula for degradable channels \cite{shor2short,smith08_3}. This simplifies the task of computing the capacity enormously and it happens that for some degradable channels the two capacities can be computed analytically.

Degradable channels form an important class of channels for which, thanks to the induced additivity properties, there is a good understanding of their quantum and private classical capacity. At the same time, the notion of a degradable channel seems to be fragile as a tiny perturbation of a degradable channel may not be degradable anymore. Furthermore, it is unknown whether a channel for which the degradability condition is approximately satisfied (up to some $\varepsilon \geq 0$ with respect to the diamond norm) is close to a degradable channel or not.
Here, we introduce a robust generalization of the concept of a degradable channel. We call a channel \emph{$\varepsilon$-degradable} if the degradability condition with respect to the diamond norm is satisfied up to some $\varepsilon \geq 0$. (The precise definition is given in Definition~\ref{def:epsiDeg}.)
We show that these $\varepsilon$-degradable channels approximately inherit all the desirable properties that degradable channels have, such as additivity of the channel's coherent and channel private information. We further show that for an arbitrary channel, the smallest $\varepsilon \geq 0$ such that the channel is $\varepsilon$-degradable can be efficiently computed via a semidefinite program. This offers a universal method to compute efficiently upper bounds to the quantum and private classical capacity. This will be demonstrated by concrete examples, including the depolarizing channel.

\vspace{2mm}
\introsection{Structure}
The remainder of this article is structured as follows. Section~\ref{sec:preliminaries} introduces a few preliminary results and gives an overview of what is known for degradable channels. Section~\ref{sec:main} presents our main contribution which is a definition of approximate degradable channels that approximately inherits all the desirable properties degradable channels have. We show that for an arbitrary channel the smallest possible $\varepsilon$ such that the channel is $\varepsilon$-degradable can be computed efficiently via a semidefinite program. Section~\ref{sec:convexUB} shows how the concept of approximate degradable channels can be used to derive powerful upper bounds to the quantum and private classical capacity. In Section~\ref{sec:examples} we discuss some examples and show that the upper bounds based on approximate degradable channels can be very tight. 
In Appendix~\ref{app:alternativeDef} we discuss an alternative definition of approximate degradability, based on closeness to a degradable channel, and compare this to our definition, based on the degradability condition being approximately satisfied. We will argue that our choice leads to a more natural and ultimately more useful notion.

\section{Preliminaries} \label{sec:preliminaries}
\introsection{Notation}
The logarithm with base 2 is denoted by $\log(\cdot)$ and the natural logarithm by $\ln(\cdot)$. For $k\in \mathbb{N}$, let $[k]:=\left \lbrace 1,\ldots,k \right \rbrace$. The space of Hermitian operators on a finite-dimensional Hilbert space $\mathcal{H}$ is denoted by $\Hop$, where $|\mathcal{H}|$ is the dimension of $\mathcal{H}$. The cone of positive semidefinite Hermitian operators of dimension $n$ is denoted by $\Hp$.
The space of trace class operators acting on some Hilbert space $\mathcal{H}$ is denoted by $\mathcal{S}(\mathcal{H})$. A quantum channel from a system $A$ to a system $B$ is represented by a completely positive trace-preserving (cptp) linear map $\Phi: \mathcal{S}(A) \to \mathcal{S}(B)$.
 For $\rho \in \mathcal{D}(\mathcal{H})$, where $\mathcal{D}(\mathcal{H}):= \{ \rho\in \Hp \, : \,  \Tr{\rho}=1 \}$ denotes the space of density operators on $\mathcal{H}$, the von Neumann entropy is defined as $H(\rho):=-\Tr{\rho \log \rho}$. For a bipartite state $\rho_{AB} \in \mathcal{D}(A \otimes B)$ we define the conditional entropy as $H(A|B)_{\rho}:=H(AB)_{\rho} - H(B)_{\rho}$ and the mutual information as $I(A:B)_{\rho}:=H(A)_{\rho}+H(B)_{\rho} - H(AB)_{\rho}$.
 For a state $\phi_{AA'} \in \mathcal{D}(A \otimes A')$ and a channel $\Phi: \mathcal{S}(A') \to \mathcal{S}(B)$ let $\rho_{AB}:=\Phi(\phi_{AA'})$. Then an equivalent characterization of the coherent information defined by $\Icc{\rho}{\Phi}:=H (\Phi(\rho))- H(\Phi^{\setC}(\rho))$, where $\Phi^{\setC}$ is a complementary channel (that is defined just below) is $I(A\rangle B)_{\rho}:=H(B)_{\rho}-H(AB)_{\rho}$. 
 For a matrix $A,B \in \mathbb{C}^{m\times n}$, we denote the Frobenius inner product by $\left \langle A,B \right \rangle_F := \Tr{A \herm B}$ and the induced Frobenius norm by $\norm{A}_F:=\sqrt{\left \langle A,A \right \rangle_F}$. The trace norm is defined as $\norm{A}_{\trnorm} :=\mathrm{tr}[\sqrt{A^{\dagger}A}]$. The operator norm is denoted by $\norm{A}_{\opnorm}:=\sup_X \{\norm{AX}_F: \, \norm{X}_F = 1\} $. 
 For a linear map $\Phi:\mathcal{S}(A)\to \mathcal{S}(B)$ its diamond norm is defined by $\norm{\Phi}_{\diamond}:=\norm{\Phi \otimes \cII_{A}}_{\trnorm}$, where $\norm{\cdot}_{\trnorm}$ denotes the trace norm for resources which is defined as $\norm{\Phi}_{\trnorm}:=\max_{\rho \in \mathcal{D}(A)}  \norm{\Phi(\rho)}_{\trnorm}$ and $\cII_A$ denotes the identity map on $A$. We denote the standard $n-$simplex by $\Delta_{n}:=\left\{  x\in\R^{n} : x\geq 0, \sum_{i=1}^{n} x_{i}=1\right\}$. The binary entropy function is defined as $\Hb(\alpha):=-\alpha \log(\alpha)-{(1-\alpha)\log(1-\alpha)}$, for $\alpha\in [0,1]$. 
 \vspace{2mm}
 
 \introsection{Quantum channels}
 A completely positive trace-preserving map $\Phi: \mathcal{S}(A) \to \mathcal{S}(B)$ can be represented in different ways. In this article we will use three different representations that are known as \emph{Stinespring}, \emph{Kraus operator}, and \emph{Choi-Jamio{\l}kowski} representation. 
 Stinespring's representation theorem \cite{stinespring55} ensures that every quantum channel $\Phi:\mathcal{S}(A) \to \mathcal{S}(B)$ can be written in terms of an isometry $V$ from $A$ to the joint system $B \otimes E$ followed by a partial trace such that $\Phi(\rho) = \Trp{E}{V \rho \,V^{\dagger}}$ for all $\rho \in \mathcal{S}(A)$. Tracing out system $B$ instead of $E$ defines a \emph{complementary channel} $\Phi^{\setC}(\rho) = \Trp{B}{V \rho \,V^{\dagger}}$ for all $\rho \in \mathcal{S}(\mathcal{H}_A)$. 
Let $|A|:=\dim A$ and $|B|:=\dim B$ denote the input and output dimension of the quantum channel, respectively, and suppose the environment has dimension $|E|:=\dim E$. 
The Kraus representation theorem ensures that for every cptp map $\Phi:\mathcal{S}(A) \to \mathcal{S}(B)$ there exists a family  $\{ F_x \}_x$ of operators $F_x$ that map from $A$ to the $B$ system, such that 
\begin{equation*}
\Phi: \mathcal{S}(\mathcal{H}_A) \ni \rho \mapsto \Phi(\rho) = \sum_{x} F_x \rho F_x^\dagger \in \mathcal{S}(\mathcal{H}_B)
\end{equation*}
and
\begin{equation*}
 \sum_x F_x^\dagger F_x = \id_A \ .
\end{equation*}
The Choi-Jamio{\l}kowski representation of the channel $\Phi:\mathcal{S}(A) \to \mathcal{S}(B)$ is the operator $J(\Phi) \in \mathcal{S}(B \otimes A)$ that is defined as
\begin{align*}
J(\Phi):=|A| \left(\Phi \otimes \cII_A \right) \left( \ketbra{\Omega}{\Omega}  \right) = \sum_{1\leq i,j \leq |A|} \Phi(E_{ij})\otimes E_{ij} \, ,
\end{align*}
where $\ket{\Omega} = \frac{1}{\sqrt{d}} \sum_{j=1}^d \ket{jj}$ denotes a maximally entangled state and $E_{ij}$ is a $(|A| \times |A|)$ matrix with a one entry at position $(i,j)$ and zeros everywhere else. It is well known that the mapping $\Phi$ is completely positive if and only if $J(\Phi) \geq 0$ and that $\Phi$ is trace-preserving if and only if $\Trp{B}{J(\Phi)}=\mathds{1}_A$. 
Using the Choi-Jamio{\l}kowski representation, the action of the channel $\Phi$ can be written as
\begin{equation*}
\Phi(\rho) = \Trp{A}{J(\Phi) \left(\id_B \otimes \rho\transp \right)}\, ,
\end{equation*}
for $\rho \in \mathcal{S}(A)$ where the transpose is with respect to the basis chosen for the maximally entangled state $\ket{\Omega}$. 

An important class of cptp maps with beneficial properties are the so-called degradable and anti-degradable channels that were introduced in \cite{shor2short}.
\begin{mydef}[Degradable and anti-degradable channels] \label{def:degraded}
A channel $\Phi:\mathcal{S}(A)\to \mathcal{S}(B)$ is called \emph{degradable} if there exists a cptp map $\Xi:\mathcal{S}(B)\to \mathcal{S}(E)$ such that $\Phi^{\setC} = \Xi \circ \Phi$. The channel $\Phi$ is called \emph{anti-degradable} if there exists a cptp map $\Theta:\mathcal{S}(E)\to \mathcal{S}(B)$ such that $\Phi = \Theta \circ \Phi^{\setC}$.
\end{mydef}

Note that $\Phi$ is anti-degradable if and only if $\Phi^\setC$ is degradable. Furthermore, the set of anti-degradable channels is convex and contains the set of entanglement-breaking channels (which is a set with positive volume under the set of all channels)~\cite{cubitt08}.

 \vspace{2mm}
 
 \introsection{Quantum and private classical capacity}
The highest rate at which quantum information can be transmitted asymptotically reliably per channel use is called quantum capacity and is mathematically characterized by the celebrated LSD formula \cite{lloyd_capacity_1997,shor_quantum_2002,devetak_private_2005} (see also~\cite{schumi96,nielsen96,schumi98,barnum00}) 
\begin{equation} \label{eq:Q}
Q(\Phi) = \lim \limits_{k \to \infty} \frac{1}{k} Q^{(1)}(\Phi^{\otimes k}) \, ,
\end{equation}
with
\begin{align}
Q^{(1)}(\Phi) &:= \max \limits_{\rho \in \mathcal{D}(A)} \Icc{\rho}{\Phi} \\ 
&:= \max \limits_{\rho \in \mathcal{D}(A)} H \bigl(\Phi(\rho)\bigr)- H\bigl( \Phi^{\setC}(\rho)\bigr) \, , \label{eq:channelCoherentInfo}
\end{align}
where $\Icc{\rho}{\Phi}$ denotes the coherent information, $Q^{(1)}(\Phi)$ is called channel coherent information, and $\Phi^{\setC}$ is a complementary channel to $\Phi$. 

Due to the regularization in \eqref{eq:Q}, the quantum capacity is difficult to compute. As a consequence, it is of interest to derive good lower and upper bounds to $Q(\Phi)$. It is immediate to verify that $Q^{(1)}(\Phi) \leq Q(\Phi)$ is valid for every channel $\Phi$, i.e., the channel coherent information is always a lower bound for the quantum capacity. However in general, this lower bound is  not tight, i.e., there exist channels $\Phi$ such that $Q^{(1)}(\Phi)<Q(\Phi)$ \cite{shorshort,smithshort}. To derive generic upper bounds for the quantum capacity that can be computed efficiently turns out to be difficult. Beside a few channel specific techniques \cite{wolf07,smith08,smith08_2} that will be discussed in Section~\ref{sec:convexUB}, generic upper bounds have been introduced based on a no-cloning argument \cite{bennett_mixed-state_1996,bruss98,cerf00} or semidefinite programming bounds \cite{rains99,rains01}. However, none of these generic upper bounds is expected to be particularly tight as explained in \cite{smith08_2}. It is thus fair to say that the quantum capacity is still poorly understood in general --- even for very low-dimensional channels. For degradable channels it has been shown that the channel coherent information is additive, i.e., that $Q(\Phi) = Q^{(1)}(\Phi)$ \cite{shor2short}. In general for a given channel $\Phi$, the function $\rho \mapsto \Icc{\rho}{\Phi}$ is not concave which complicates the task of computing $Q^{(1)}(\Phi)$ defined in \eqref{eq:channelCoherentInfo}. However, $\Phi$ being degradable implies that $\rho \mapsto \Icc{\rho}{\Phi}$ is concave \cite[Lem.~5]{yard08} and as such $Q^{(1)}(\Phi)$, and hence $Q(\Phi)$, is characterized via a finite-dimensional convex optimization problem. We also note that due to a no-cloning argument, anti-degradable channels must have a zero quantum capacity, i.e., $Q(\Phi)=0$ \cite{bennett_mixed-state_1996,holevo_book}.

The private classical capacity of a quantum channel characterizes the highest possible rate at which classical information  can be transmitted asymptotically reliably per channel use such that no information about the message leaks to the environment.
It is mathematically characterized by the regularized private channel information  \cite{devetak_private_2005,cai04}, i.e.,
\begin{equation} \label{eq:P}
P(\Phi)=\lim \limits_{k\to \infty} \frac{1}{k} P^{(1)}(\Phi^{\otimes k})\, ,
\end{equation}
with channel private information
\begin{align*}
P^{(1)}(\Phi):= \max_{\{\rho_i,p_i \}} \Big \{ \Hh{\sum_i p_i \Phi(\rho_i)}- \sum_i p_i H\bigl(\Phi(\rho_i)\bigr)-\Hh{\sum_i p_i \Phi^{\setC}(\rho_i)}+ \sum_i p_i H\bigl(\Phi^{\setC}(\rho_i)\bigr) \Big\}\ .
\end{align*}
Similar to the quantum capacity, the regularization arising in \eqref{eq:P} complicates the task of evaluating the private classical capacity and the channel private information $P^{(1)}(\Phi)$ is always a lower bound to $P(\Phi)$ which however in general is not tight \cite{renes08}. Finding generic upper bounds to the private classical capacity again turns out to be difficult. For degradable channels it has been shown that $P^{(1)}(\Phi)=P(\Phi)$, whereas for anti-degradable channels $P(\Phi)=0$ holds \cite{smith08_3}.

There is a close connection between the quantum capacity and the private classical capacity of a quantum channel. Since fully quantum communication is necessarily private, $Q(\Phi) \leq P(\Phi)$ for every channel $\Phi$. It can further be shown that $Q^{(1)}(\Phi) \leq P^{(1)}(\Phi)$ for all channels $\Phi$~\cite{cai04,devetak_private_2005} (see also~\cite[Thm.~12.6.3]{wilde_book}). For degradable channels we have $Q(\Phi) = Q^{(1)}(\Phi) = P^{(1)}(\Phi)=P(\Phi)$ \cite{shor2short,smith08_3}, while for anti-degradable channels we have seen that $Q(\Phi)=P(\Phi)=0$ \cite{holevo_book,smith08_3}.

In the following we will oftentimes use two continuity results for the von Neumann entropy. The first one is known as the \emph{Fannes-Audenaert inequality}.
\begin{mylem}[\cite{audenaert07}] \label{lem:fannesAudenaert}
For any states $\rho \in \mathcal{D}(A)$ and $\sigma \in \mathcal{D}(A)$ such that $\frac{1}{2}\norm{\rho - \sigma}_{\trnorm} \leq \varepsilon \leq 1$,
\begin{align*}
|H(\rho) - H(\sigma)| \leq \varepsilon \log(|A|-1) + \Hb(\varepsilon) \, .
\end{align*}
\end{mylem}
The second statement we will oftentimes use is a continuity result of the conditional von Neumann entropy which is a recent strengthening of the \emph{Alicki-Fannes inequality}~\cite{alicki04} by one of the authors (AW), partly motivated by a first version of this contribution. 
\begin{mylem}[{\cite[Lemma~2]{winter15}}] \label{lem:alickiFannes}
For any states $\rho_{AB} \in \mathcal{D}(A \otimes B)$ and $\sigma_{AB} \in\mathcal{D}(A \otimes B)$ such that $\frac{1}{2}\norm{\rho_{AB} - \sigma_{AB}}_{\trnorm} \leq \varepsilon \leq 1$, 
\begin{equation*}
| H(A|B)_{\rho} - H(A|B)_{\sigma} | \leq 2\varepsilon \log(|A|)  + \Bigl(1+\varepsilon \Bigr) \, \Hb\Bigl(\frac{\varepsilon}{1+\varepsilon}\Bigr) \, .
\end{equation*}
\end{mylem}

\section{Approximate degradable channels} \label{sec:main}
In this section we precisely define the concept of an approximate degradable channel. Theorem~\ref{thm:epsilonDegradableAdditivity} then shows that the desirable additivity properties of degradable channels are approximately inherited by $\varepsilon$-degradable channels. Proposition~\ref{prop:SDP} shows that the smallest possible $\varepsilon$ such that a given channel is $\varepsilon$-degradable can be efficiently computed via a semidefinite program. Finally we show that in the same spirit we can also define $\varepsilon$-anti-degradable channels (as done in Definition~\ref{def:epsiAntiDeg}) which approximately inherit the properties of anti-degradable channels (cf.\ Theorem~\ref{thm:antideg}).

\begin{mydef} [$\varepsilon$-degradable] \label{def:epsiDeg}
A channel $\Phi:\mathcal{S}(A) \to \mathcal{S}(B)$ is said to be \emph{$\varepsilon$-degradable} if there exists another channel $\Xi:\mathcal{S}(B) \to \mathcal{S}(E)$ such that $\norm{\Phi^{\setC}-\Xi \circ \Phi}_{\diamond} \leq \varepsilon$.
\end{mydef}
According to Definition~\ref{def:epsiDeg}, we call a channel \emph{$\varepsilon$-degradable} if the degradability condition is approximately satisfied. We note that a possible alternative characterization of approximate degradable channels is to call a channel \emph{$\varepsilon$-close-degradable} if it is close to a degradable channel  (cf.\ Definition~\ref{def:EpsiCloseDeg}). Within this article, we will focus on the first definition of approximate degradable channels. The latter approach is discussed in Appendix~\ref{app:alternativeDef} where we also mention major differences between these two definitions of approximate degradable channels which will justify why we favor $\varepsilon$-degradable channels over $\varepsilon$-close-degradable channels. 

By Definition~\ref{def:epsiDeg}, every channel is $\varepsilon$-degradable for some $\varepsilon \in [0,2]$ since $\norm{\Phi^{\setC}-\Xi \circ \Phi}_{\diamond} \leq \norm{\Phi^{\setC}}_{\diamond} +\norm{\Xi \circ \Phi}_{\diamond} \leq 2$.
The following theorem (Theorem~\ref{thm:epsilonDegradableAdditivity}) ensures that $\varepsilon$-degradable channels inherit the desirable additivity properties of the channel coherent and the channel private information that degradable channels offer, up to an error term that vanishes in the limit $\varepsilon \to 0$.

For a channel $\Phi$ from $A$ to $B$ and a degrading channel $\Xi$ from $B$ to $\widetilde{E}\simeq E$, choose Stinespring isometric dilations 
$V:A\hookrightarrow B\otimes E$ and $W:B\hookrightarrow \widetilde{E}\otimes F$, respectively. Then, define
\begin{align} 
  &U_{\Xi}(\Phi):=\max_{\rho \in \mathcal{D}(A)} \{ H(F|\widetilde{E})_\omega \, :  \ \omega_{E\widetilde{E}F}  = (W\otimes\1)V\rho V^\dagger(W\otimes\1)^\dagger \} \, .\label{eq:U}
\end{align}
\begin{myprop}
  \label{prop:Q1-vs-U}
  If $\Phi:\mathcal{S}(A) \to \mathcal{S}(B)$ is an $\varepsilon$-degradable channel with a degrading channel $\Xi:\mathcal{S}(B) \to \mathcal{S}(E)$, then
  \begin{equation*}
    \bigl| Q^{(1)}(\Phi) - U_{\Xi}(\Phi) \bigr| 
         \leq \frac{\varepsilon}{2}\log(|E|-1) + h\left(\frac{\varepsilon}{2}\right) \, .
  \end{equation*}
\end{myprop}
  Note that for degradable $\Phi$, i.e., $\varepsilon=0$, this reproduces the main observation of Devetak and Shor~\cite{shor2short}. The significance of $U_{\Xi}(\Phi)$ is hence that it approximates $Q^{(1)}(\Phi)$ --- with equality for $\varepsilon=0$~\cite{shor2short} ---, and at the same time it is given by a convex optimization problem, which in general considerably simplifies the task of computing it~\cite{ref:BoyVan-04}. Furthermore, $U_{\Xi}(\Phi)$ itself is additive.
\begin{mylem}
Let $\Phi_1, \Phi_2 :\mathcal{S}(A) \to \mathcal{S}(B)$ be channels and $\Xi_1,\Xi_2:\mathcal{S}(A) \to \mathcal{S}(E)$ be degrading channels. Then
\begin{equation}
  \label{eq:U-additive}
  U_{\Xi_1\otimes\Xi_2}(\Phi_1\otimes\Phi_2)
      = U_{\Xi_1}(\Phi_1) + U_{\Xi_2}(\Phi_2)\ .
\end{equation}
\end{mylem}
 \begin{proof}
This proof follows the original argument of Devetak and Shor for degradable channels~\cite[App.~B]{shor2short}. Indeed, first, for channels $\Phi_i$ and degrading channels $\Xi_i$ ($i=1,2$), we have
\begin{equation*}
  U_{\Xi_1\otimes\Xi_2}(\Phi_1\otimes\Phi_2)
      \geq U_{\Xi_1}(\Phi_1) + U_{\Xi_2}(\Phi_2)\, ,
\end{equation*}
by choosing a product input state.
Second, for a state $\rho_{A_1A_2}$ and the corresponding
$\omega_{E_1\widetilde{E}_1F_1 E_2\widetilde{E}_2 F_2}$,
\begin{align*}
  H(F_1F_2|\widetilde{E}_1\widetilde{E}_2) 
  &\leq H(F_1|\widetilde{E}_1\widetilde{E}_2) + H(F_2|\widetilde{E}_1\widetilde{E}_2) \\
  &\leq H(F_1|\widetilde{E}_1) + H(F_2|\widetilde{E}_2)\, ,
\end{align*}
by applying strong subadditivity\footnote{Recall that the celebrated strong subadditivity of quantum entropy~\cite{LieRus73_1,LieRus73} ensures that for any state $\rho_{ABC}$ we have $H(A|B)_{\rho} \geq H(A|BC)_{\rho}$.} three times. Hence maximizing over input states
we get~\eqref{eq:U-additive}.
 \end{proof}

\begin{proof}[Proof of Proposition~\ref{prop:Q1-vs-U}]
The single-letter coherent information is
\begin{align*}
  Q^{(1)}(\Phi) &= \max_{\rho \in \mathcal{D}(A)} H\bigl(\Phi(\rho)\bigr) - H\bigl(\Phi^\setC(\rho)\bigr)\\
                       &= \max_{\rho\in \mathcal{D}(A)} H(\widetilde{E}F)_{\omega}-H(E)_{\omega} \, ,
\end{align*}
with respect to the state $\omega$ introduced in~\eqref{eq:U}. Since by assumption $\Phi$ is $\varepsilon$-degradable, 
${\|\omega_E-\omega_{\widetilde{E}}\|_1} \leq \varepsilon$, and invoking the Fannes-Audenaert inequality (see Lemma~\ref{lem:fannesAudenaert}) to
replace $H(E)$ by $H(\widetilde{E})$, the claim follows.
\end{proof}

\begin{mythm}[Properties of $\varepsilon$-degradable channels] \label{thm:epsilonDegradableAdditivity}
If $\Phi:\mathcal{S}(A) \to \mathcal{S}(B)$ is an $\varepsilon$-degradable channel with a degrading channel $\Xi:\mathcal{S}(B) \to \mathcal{S}(E)$, then
\begin{enumerate}[(i)]
\item $Q^{(1)}(\Phi) \leq Q(\Phi) \leq Q^{(1)}(\Phi) +  \frac{\varepsilon}{2}  \log(|E|-1) +  \Hb(\frac{\varepsilon}{2}) +\varepsilon \log |E| + {\bigl(1+\frac{\varepsilon}{2} \bigr)} \Hb\bigl(\frac{\varepsilon}{2+\varepsilon} \bigr) $ , \label{it:i}
\item $Q(\Phi) \leq  U_{\Xi}(\Phi) + \varepsilon \log |E| +  \bigl( 1+\frac{\varepsilon}{2} \bigr) h\bigl(\frac{\varepsilon}{2+\varepsilon}\bigr)$ , \label{it:i2}
\item $P^{(1)}(\Phi) \leq P(\Phi)\leq P^{(1)}(\Phi) + \frac{\varepsilon}{2}  \log(|E|-1)  +   \Hb(\frac{\varepsilon}{2}) +  3 \varepsilon \log |E|  +  3(1+\frac{\varepsilon}{2}) \Hb(\frac{\varepsilon}{2+\varepsilon}) $  , \label{it:iii}
\item $Q^{(1)}(\Phi) \leq P^{(1)}(\Phi)\leq Q^{(1)}(\Phi)  +  \frac{\varepsilon}{2}  \log(|E|-1) +  \Hb(\frac{\varepsilon}{2}) + \varepsilon \log |E| +  \bigl(1+\frac{\varepsilon}{2} \bigr) \Hb\bigl(\frac{\varepsilon}{2+\varepsilon} \bigr)$\, . \label{it:ii}
\end{enumerate}
\end{mythm}

By combining the four statements given in the theorem above we can generate other interesting upper bounds such as
\begin{align*}
P(\Phi) 
& \leq P^{(1)}(\Phi)+  \frac{\varepsilon}{2}  \log(|E|-1) +   \Hb\left(\frac{\varepsilon}{2}\right) + 3 \varepsilon \log |E| + 3\left(1+\frac{\varepsilon}{2}\right) \Hb\left(\frac{\varepsilon}{2+\varepsilon}\right) \\
&\leq Q^{(1)}(\Phi) +  \varepsilon \log(|E|-1) +   2 \Hb\left(\frac{\varepsilon}{2}\right) + 4 \varepsilon \log |E| + 4\left(1+\frac{\varepsilon}{2}\right) \Hb\left(\frac{\varepsilon}{2+\varepsilon}\right) \\
&\leq Q(\Phi) +  \frac{3\varepsilon}{2} \log(|E|-1) +   3 \Hb\left(\frac{\varepsilon}{2}\right) + 5 \varepsilon \log |E| + 5\left(1+\frac{\varepsilon}{2}\right) \Hb\left(\frac{\varepsilon}{2+\varepsilon}\right) \\
&\leq U_{\Xi}(\Phi) +  \frac{3\varepsilon}{2} \log(|E|-1) +   3 \Hb\left(\frac{\varepsilon}{2}\right) + 6 \varepsilon \log |E| + 6\left(1+\frac{\varepsilon}{2}\right) \Hb\left(\frac{\varepsilon}{2+\varepsilon}\right) \, ,
\end{align*}
where we used statements~\eqref{it:iii},~\eqref{it:ii},~\eqref{it:i}, and~\eqref{it:i2}.

We note that there exist different ways to prove the statements~\eqref{it:i}-\eqref{it:ii} of Theorem~\ref{thm:epsilonDegradableAdditivity}. On the one hand, they follow easily by exploiting the properties of the quantity \(U_{\Xi}(\Phi)\). On the other hand, they could also be shown by a telescoping sum technique, which we here use to prove the statement~\eqref{it:iii}. A similar result would also follow by combining~\eqref{it:i2} and~\eqref{it:ii}, but would give slightly worse bounds. To prove statement~\eqref{it:iii}, we need two preliminary lemmas. To simplify notation, let us define a multivariate mutual information for $n$ quantum systems $A_1,\dots,A_n$ and a state $\rho_{A_1 A_2 \ldots A_n} \in \mathcal{D}(\otimes_{i=1}^n A_i)$ as $I(A_1:A_2:\ldots:A_n)_{\rho}:=\sum_{i=1}^nH(A_i)_{\rho} - \Hh{A_1 A_2 \ldots A_n}_{\rho}$.
\begin{mylem} \label{lem:help}
Let $\varepsilon \geq 0$, $\Phi:\mathcal{S}(A)\to \mathcal{S}(B)$ and $\Xi:\mathcal{S}(B)\to \mathcal{S}(E)$ be two channels such that $\norm{\Phi^{\setC}-\Xi\circ \Phi}_{\diamond} \leq \varepsilon$. Let $j\in[n]$ and consider states $\rho_{A_1\ldots A_n}$, $\sigma_{A_1 \ldots A_{j-1} B_j A_{j+1} \ldots A_n}:=\Phi_{A_j \to B_j}(\rho_{A_1\ldots A_n})$, $\eta_{A_1 \ldots A_{j-1} E_j A_{j+1} \ldots A_n}:=\Phi^{\setC}_{\,A_j \to E_j}(\rho_{A_1\ldots A_n})$ and $\nu_{A_1 \ldots A_{j-1} E_j A_{j+1} \ldots A_n}:=\Xi_{B_j \to E_j} \circ \Phi_{A_j \to B_j}(\rho_{A_1\ldots A_n})=\Xi_{B_j \to E_j}(\sigma_{A_1 \ldots A_{j-1} B_j A_{j+1} \ldots A_n})$. Then,
\begin{align*}
&I(A_1\!:\!\ldots\!:\!A_{j-1}\!:\!B_j\!:\!A_{j+1}\!:\!\ldots:A_n)_{\sigma}  - I(A_1\!:\!\ldots\!:\!A_{j-1}\!:\!E_j\!:\!A_{j+1}\!:\!\ldots:A_n)_{\eta}\\
&\hspace{20mm} \geq- \frac{\varepsilon}{2}  \log(|E|\!-\!1) \!-\!  \Hb\Bigl(\frac{\varepsilon}{2}\Bigr)  \!-\!\varepsilon \log(|E|) \!-\! \Bigl(1\!+\!\frac{\varepsilon}{2} \Bigr) \Hb\Bigl(\frac{\varepsilon}{2\!+\!\varepsilon} \Bigr) \, .
\end{align*}
\end{mylem}
\begin{proof}
By definition of the multivariate mutual information and since $\Hh{A_i}_{\sigma}=\Hh{A_i}_{\eta}$ for $i\ne j$ we obtain
\begin{align*}
&I(A_1:\ldots:A_{j-1}:B_j:A_{j+1}:\ldots:A_n)_{\sigma} - I(A_1:\ldots:A_{j-1}:E_j:A_{j+1}:\ldots:A_n)_{\eta} \\
&\hspace{5mm} = \Hh{B_j}_{\sigma}-\Hh{E_j}_{\eta}- \Hh{A_1\ldots A_{j-1}B_jA_{j+1}\ldots A_n}_{\sigma} +\Hh{A_1\ldots A_{j-1}E_jA_{j+1}\ldots A_n}_{\eta} \\
&\hspace{5mm} =  \Hh{B_j}_{\sigma}-\Hh{E_j}_{\eta} - \Hc{B_j}{A_1 \ldots A_{j-1} A_{j+1} \ldots A_n}_{\sigma} + \Hc{E_j}{A_1 \ldots A_{j-1} A_{j+1} \ldots A_n}_{\eta}\\
&\hspace{5mm} \geq \Hh{B_j}_{\sigma}-\Hh{E_j}_{\nu} - \frac{\varepsilon}{2}  \log(|E|-1)  - \Hc{B_j}{A_1 \ldots A_{j-1} A_{j+1} \ldots A_n}_{\sigma} -  \Hb(\varepsilon/2) \nonumber \\
&\hspace{15mm}+ \Hc{E_j}{A_1 \ldots A_{j-1} A_{j+1} \ldots A_n}_{\nu} -\varepsilon \log |E| - \Bigl(1+\frac{\varepsilon}{2} \Bigr) \Hb\Bigl(\frac{\varepsilon}{2+\varepsilon} \Bigr)\\
&\hspace{5mm} =\I{B_j}{A_1 \ldots A_{j-1} A_{j+1} \ldots A_n}_{\sigma}-\I{E_j}{A_1 \ldots A_{j-1} A_{j+1} \ldots A_n}_{\nu} - \frac{\varepsilon}{2}  \log(|E|-1)\\
&\hspace{15mm} -  \Hb(\varepsilon/2)   -\varepsilon \log |E| - \Bigl(1+\frac{\varepsilon}{2} \Bigr) \Hb\Bigl(\frac{\varepsilon}{2+\varepsilon} \Bigr)\\
&\hspace{5mm} \geq - \frac{\varepsilon}{2}  \log(|E|-1) -  \Hb(\varepsilon/2)  -\varepsilon \log |E|  - \Bigl(1+\frac{\varepsilon}{2} \Bigr) \Hb\Bigl(\frac{\varepsilon}{2+\varepsilon} \Bigr)\, ,
\end{align*}
where the second equality uses that 
\begin{align*}
&\Hh{A_1 \ldots A_{j-1} A_{j+1} \ldots A_n}_{\sigma} = \Hh{A_1 \ldots A_{j-1} A_{j+1} \ldots A_n}_{\eta} \, .
\end{align*}
The first inequality follows from Fannes-Audenaert inequality (cf.~Lemma~\ref{lem:fannesAudenaert}) and the Alicki-Fannes inequality (cf.~Lemma~\ref{lem:alickiFannes}). The final inequality follows from monotonicity of mutual information under local channels, which is also known as the data processing inequality for mutual information.
\end{proof}

\begin{proof}[Proof of Theorem~\ref{thm:epsilonDegradableAdditivity}]  \
We start by proving claim~\eqref{it:i2}. For $n$ uses of the channel $\Phi$, consider an input
state $\rho$ on $A^n = A_1\ldots A_n$ and the corresponding
$\omega$ on $E^n\widetilde{E}^nF^n$. Then, the coherent information of the
channel equals
\begin{align*}
  H(\widetilde{E}^nF^n) - H(E^n) 
  &= H(\widetilde{E}^nF^n) - H(\widetilde{E}^n) + H(\widetilde{E}^n) - H(E^n) \\
  &= H(F^n|\widetilde{E}^n) + \sum_{t=1}^n \bigl[ H(\widetilde{E}_t|\widetilde{E}_{<t}E_{>t})-H(E_t|\widetilde{E}_{<t}E_{>t})\bigr] \, ,
\end{align*}
``telescoping'' the difference into a sum of ``local'' differences.
Each term in the latter sum has its modulus bounded by 
$\delta := \varepsilon \log|E| + \left( 1+\frac{\varepsilon}{2} \right) h\!\left(\frac{\varepsilon}{2+\varepsilon}\right)$,
via Lemma~\ref{lem:alickiFannes}.
Maximizing over input states and using the additivity of 
$U_{\Xi}(\Phi)$, Equation~\eqref{eq:U-additive}, yields
\[
  Q^{(1)}\bigl(\Phi^{\otimes n}\bigr) \leq n\, U_{\Xi}(\Phi) + n \delta \, .
\]
Dividing by $n$ and taking the limit $n\rightarrow\infty$ gives the statement of claim~\eqref{it:i2}.

We next prove claim~\eqref{it:i}. The lower bound $Q^{(1)}(\Phi)\leq Q(\Phi)$ is immediate. Combining claim~\eqref{it:i2} with Proposition~\ref{prop:Q1-vs-U} immediately proves claim~\eqref{it:i}. Appendix~\ref{app:alternativeProof} presents an alternative proof for claim~\eqref{it:i}.

We next prove claim~\eqref{it:iii}.
Note that $P^{(1)}(\Phi) \leq P(\Phi)$ for an arbitrary channel $\Phi$ is proven in \cite[p.~323]{wilde_book}, therefore only the upper bound on $P(\Phi)$ needs to be shown. Let $\xi_1:=2 \varepsilon \log |E| + 2(1+\frac{\varepsilon}{2}) \Hb(\frac{\varepsilon}{2+\varepsilon})$, $\xi_2:=\frac{\varepsilon}{2}  \log(|E|-1) +  \Hb(\varepsilon/2)  +\varepsilon \log |E| + \Bigl(1+\frac{\varepsilon}{2} \Bigr) \Hb\Bigl(\frac{\varepsilon}{2+\varepsilon} \Bigr)$ and suppose $\rho_{X A'_1 \ldots A'_n}$ is the state that maximizes $P^{(1)}(\Phi^{\otimes n})$, where 
\begin{equation*}
\rho_{X A'_1 \ldots A'_n}:= \sum_{x \in \mathcal{X}} P_X(x) \braket{x}{x}_{X} \otimes \rho^x_{A'_1 \ldots A'_n} \, ,
\end{equation*}
and let 
\begin{align*}
\rho_{X B_1 E_1 \ldots B_n E_n} := \left( \bigotimes_{i=1}^n V^{i}_{A_i \to B_i E_i} \right) \rho_{X A'_1 \ldots A'_n} \left( \bigotimes_{i=1}^n V^{i}_{A_i \to B_i E_i } \right)^{\dagger}
\end{align*}
be the state that arises when sending $\rho_{X A'_1 \ldots A'_n}$ through $\Phi^{\otimes n}$, with $V^{i}_{A_i \to B_i E_i}$ denoting the isometric extension of the $i$-th channel $\Phi$. Consider a spectral decomposition of each state $\rho^{x}_{A'_1 \ldots A'_n}$  as $\rho^{x}_{A'_1 \ldots A'_n} = \sum_{y \in \mathcal{Y}} P_{Y|X}(y|x) \varphi^{x,y}_{A'_1 \ldots A'_n}$, where each state $\varphi^{x,y}_{A'_1 \ldots A'_n}$ is pure. Let $\sigma_{X Y A'_1 \ldots A'_n}$ be an extension of $\rho_{X A'_1 \ldots A'_n}$ with
\begin{align*}
\sigma_{XYA'_1 \ldots A'_n}:= \sum_{x \in \mathcal{X},\, y \in \mathcal{Y}} P_{Y|X}(y|x) P_X(x) \braket{x}{x}_X \otimes \braket{y}{y}_{Y} \otimes \varphi^{x,y}_{A'_1 \ldots A'_n}
\end{align*}
and let 
\begin{align*}
\sigma_{XYB_1 E_1 \ldots B_n E_n} :=  \left( \bigotimes_{i=1}^n V^{i}_{A_i \to B_i E_i }\right) \sigma_{XYA'_1 \ldots A'_n} \left( \bigotimes_{i=1}^n V^{i}_{A_i \to B_i E_i}  \right)^{\dagger}
\end{align*}
be the state that arises when sending $\sigma_{XYA'_1 \ldots A'_n}$ through $\Phi^{\otimes n}$. By assumption 
\begin{align}
P^{(1)}(\Phi^{\otimes n}) 
&\hspace{0mm}= \I{X}{B_1 \ldots B_n}_{\rho} - \I{X}{E_1 \ldots E_n}_{\rho} \nonumber\\
&\hspace{0mm}= \I{X}{B_1 \ldots B_n}_{\sigma} - \I{X}{E_1 \ldots E_n}_{\sigma} \label{eq:dpure}\\
&\hspace{0mm}=\I{XY}{B_1 \ldots B_n}_{\sigma} \!- \! \I{XY}{E_1 \ldots E_n}_{\sigma} \nonumber \\
&\hspace{10mm}-\bigl(\Ic{Y}{B_1 \ldots B_n}{X}_{\sigma}\!-\!\Ic{Y}{E_1 \ldots E_n}{X}_{\sigma}  \bigr)\nonumber\\
&\hspace{0mm} \leq \I{XY}{B_1 \ldots B_n}_{\sigma} \!- \! \I{XY}{E_1 \ldots E_n}_{\sigma} + n \xi_1 \label{eq:fINE}\\
&\hspace{0mm}=\Hh{B_1 \ldots B_n}_{\sigma}- \Hc{B_1 \ldots B_n}{XY}_{\sigma} \nonumber\\ 
&\hspace{10mm}- \Hh{E_1 \ldots E_n}_{\sigma} + \Hc{E_1 \dots E_n}{XY}_{\sigma} + n \xi_1 \nonumber\\
&\hspace{0mm}=\Hh{B_1 \ldots B_n}_{\sigma}- \Hc{B_1 \ldots B_n}{XY}_{\sigma}\nonumber \\
&\hspace{10mm} - \Hh{E_1 \ldots E_n}_{\sigma} + \Hc{B_1 \dots B_n}{XY}_{\sigma} + n \xi_1 \label{eq:dvv}\\
&\hspace{0mm}=\Hh{B_1 \ldots B_n}_{\sigma} - \Hh{E_1 \ldots E_n}_{\sigma} + n \xi_1 \nonumber\\
&\hspace{0mm}= \sum_{i=1}^n \bigl( \Hh{B_i}_{\sigma} - \Hh{E_i}_{\sigma} \bigr) - \bigl(I(B_1: \ldots :B_n)_{\sigma}- I(E_1: \ldots :E_n)_{\sigma} \bigr) + n \xi_1 \nonumber \\
&\hspace{0mm}\leq \sum_{i=1}^n \bigl( \Hh{B_i}_{\sigma} - \Hh{E_i}_{\sigma} \bigr) +  n( \xi_1 + \xi_2) \label{ineq:two} \\
&\hspace{0mm}\leq n Q^{(1)}(\Phi) + n(\xi_1 + \xi_2) \label{ineq:three} \\
&\hspace{0mm}\leq n P^{(1)}(\Phi) + n (\xi_1 + \xi_2) \label{ineq:finally}\, , 
\end{align}
where \eqref{eq:dpure} is valid since $\rho_{X B_1 E_1 \ldots B_n E_n} = \Trp{Y}{\sigma_{XYB_1 E_1 \ldots B_n E_n}}$. Inequality \eqref{eq:fINE} follows by applying $n$ times in sequence the strengthened Alicki-Fannes inequality (cf.\ Lemma~\ref{lem:alickiFannes}) for the quantum mutual information, followed by the data processing inequality together with the fact that $X$ is classical and the assumption that $\norm{\Phi^\setC - \XiÊ\circ \Phi}_{\diamond} \leq \varepsilon$. More precisely, for $i \in [n]$ let
\begin{equation*}
\eta^{i}_{XY E_1 \ldots E_n} := \left( \Bigl( \bigotimes_{j=1}^{n-1} \Phi^\setC \Bigr) \otimes \Bigl( \bigotimes_{k=1}^i \xi \circ \Phi \Bigr) \right) (\sigma_{XYA'_1\ldots A'_n}) \, .
\end{equation*}
By assumption, for all $i=2,\ldots,n$ we have $\norm{\eta_i - \eta_{i-1}}_\trnorm \leq \varepsilon$. Thus
\begin{align}
&|\Ic{Y}{E_1 \ldots E_n}{X}_{\eta_i}- \Ic{Y}{E_1 \ldots E_n}{X}_{\eta_{i-1}}|  \nonumber \\
&\hspace{10mm}= |\Ic{Y}{E_i}{X,E_1\ldots E_{i-1} E_{i+1} \ldots E_n}_{\eta^i} -  \Ic{Y}{E_i}{X,E_1\ldots E_{i-1} E_{i+1} \ldots E_n}_{\eta^{i-1}}| \label{eq:4f}\\
& \hspace{10mm} \leq \xi_1 \, , \label{eq:fff}
\end{align}
where \eqref{eq:4f} follows since by construction the states $\eta_i$ and $\eta_{i-1}$ differ only on subsystem $E_i$ and \eqref{eq:fff} is a consequence of applying the strengthened Alicki-Fannes inequality (cf.\ Lemma~\ref{lem:alickiFannes}) twice. Applying the argument described by \eqref{eq:4f} and \eqref{eq:fff} $n$ times in sequence shows that
\begin{align}
&-\Ic{Y}{B_1 \ldots B_n}{X}_{\sigma}+\Ic{Y}{E_1 \ldots E_n}{X}_{\sigma} \nonumber \\
&\hspace{20mm} \leq - \Ic{Y}{B_1 \ldots B_n}{X}_{\sigma} + \Ic{Y}{E_1 \ldots E_n}{X}_{\eta^n} + n \xi_1 \\
&\hspace{20mm}   \leq n \xi_1\, ,
\end{align}
where the final step uses the data processing inequality.
 Equation \eqref{eq:dvv} holds as $\sigma$ is pure on $B_1 E_1 \ldots B_n E_n$ when conditioning on $XY$. The second inequality \eqref{ineq:two} follows by $n$ times applying Lemma~\ref{lem:help} and \eqref{ineq:three} holds as $\rho$ is not necessarily that state maximizing the coherent information. Finally, \eqref{ineq:finally} follows from the fact that $Q^{(1)}(\Phi) \leq P^{(1)}(\Phi)$ is true for all channels $\Phi$ \cite[Theorem~12.6.3]{wilde_book}. This then proves claim~\eqref{it:iii}.

We finally prove claim \eqref{it:ii} of Theorem~\ref{thm:epsilonDegradableAdditivity}. Note that $Q^{(1)}(\Phi) \leq P^{(1)}(\Phi)$ is true for all channels $\Phi$ \cite[Theorem~12.6.3]{wilde_book}, hence only the upper bound for $P^{(1)}(\Phi)$ needs to be shown. Consider a classical-quantum state $\rho_{X A'}= \sum_{x \in \mathcal{X}} P_{X}(x) \braket{x}{x}_X \otimes \rho^x_{A'}$ and let $\sigma_{XBE}=V_{A' \to BE} \rho_{X A'} (V_{A' \to BE})^{\dagger}$, where $V_{A' \to BE}$ is the isometric extension of $\Phi$. Each state $\rho_{A'}^x$ can be decomposed as $\rho^x_{A'} = \sum_{y \in \mathcal{Y}} P_{Y|X}(y|x) \varphi^{x,y}_{A'}$, where each state $\varphi^{x,y}_{A'}$ is pure. Consider the following extension of the state $\sigma_{XBE}$
\begin{align*}
\eta_{XYBE}:= \sum_{x \in \mathcal{X},\, y \in \mathcal{Y}} \!\!\!\!\! P_{Y|X}(y|x) P_{X}(x) \braket{x}{x}_{X} \otimes \braket{y}{y}_{Y} \otimes V_{A' \to BE} \varphi_{A'}^{x,y} (V_{A' \to BE})^{\dagger}\, .
\end{align*}
Suppose that $\sigma_{XBE}$ maximizes the private information and let $\xi:= \frac{\varepsilon}{2}  \log(|E|-1)+  \Hb(\frac{\varepsilon}{2})  +\varepsilon \log |E| + \bigl(1+\frac{\varepsilon}{2} \bigr) \Hb\bigl(\frac{\varepsilon}{2+\varepsilon} \bigr)$, then
\begin{align}
P^{(1)}(\Phi)&=\I{X}{B}_{\sigma} - \I{X}{E}_{\sigma}\nonumber \\
&= \I{X}{B}_{\eta} - \I{X}{E}_{\eta} \label{eq:twoo}\\
&= \I{XY}{B}_{\eta} - \Ic{Y}{B}{X}_{\eta} - \I{XY}{E}_{\eta} + \Ic{Y}{E}{X}_{\eta}\label{eq:Cain}\\
&=\I{XY}{B}_{\eta}-\I{XY}{E}_{\eta}  - \left(\Ic{Y}{B}{X}_{\eta} - \Ic{Y}{E}{X}_{\eta} \right)\nonumber\\
&\leq \I{XY}{B}_{\eta}-\I{XY}{E}_{\eta} + \xi \nonumber\\
&= \Hh{B}_{\eta}-\Hc{B}{XY}_{\eta}- \Hh{E}_{\eta} +\Hc{E}{XY}_{\eta} + \xi\nonumber\\
&=\Hh{B}_{\eta}-\Hc{B}{XY}_{\eta}- \Hh{E}_{\eta} +\Hc{B}{XY}_{\eta} + \xi \label{eq:purr}\\
&= \Hh{B}_{\eta}- \Hh{E}_{\eta} + \xi \nonumber\\
&\leq Q^{(1)}(\Phi) + \xi\, ,\nonumber
\end{align}
where \eqref{eq:twoo} follows since $\sigma_{XBE}=\Trp{Y}{\eta_{XYBE}}$ and \eqref{eq:Cain} is a simple application of the chain rule. The inequality step follows from Lemma~\ref{lem:help} for $n=2$ and since $X$ is a classical system. Equation \eqref{eq:purr} is true since $\eta$ is pure on $EB$ when conditioning on $XY$. The final inequality follows since $\eta$ is not necessarily the state that maximize the coherent information. This proves statement \eqref{it:ii} of Theorem~\ref{thm:epsilonDegradableAdditivity}.

\end{proof}

By Definition~\ref{def:epsiDeg} it can be verified immediately that if a channel $\Phi$ is $\varepsilon$-degradable it is also $\varepsilon'$-degradable for all $\varepsilon'\geq \varepsilon$. The smallest possible parameter $\varepsilon$ such that $\Phi$ is $\varepsilon$-degradable is given by
\begin{align} \label{opt:epsi}
\begin{split}
\varepsilon_{\Phi}:= \,&\inf\limits_{\Xi} \, 		 \norm{\Phi^{\setC}-\Xi\circ \Phi}_{\diamond}\\
& \st \,\, \Xi:\mathcal{S}(B) \to \mathcal{S}(E) \textnormal{ is cptp} \ .
\end{split}
\end{align}
\begin{myprop} \label{prop:SDP}
The optimization problem \eqref{opt:epsi} can be expressed as a semidefinite program.
\end{myprop}
\begin{proof}
Watrous proved \cite[Sec.\ 4]{watrous09} that for two channels $\Theta_1, \Theta_2 : \mathcal{S}(A) \to S(B)$ the diamond norm of their difference, i.e., $\norm{\Theta_1-\Theta_2}_{\diamond}$ can be expressed as a semidefinite program (SDP) of the form
\begin{align}  \label{eq_wat}
\begin{split}
\norm{\Theta_1-\Theta_2}_{\diamond} = &2 \inf\limits_{Z} \, \norm{\Trp{B}{Z}}_{\infty} \\
& \st \,  Z \geq J(\Theta_1-\Theta_2) \\
& \hspace{7.5mm} Z \geq 0 \, ,
\end{split}
\end{align}
where $J(\Theta_1 - \Theta_2)$ denotes the Choi-Jamio{\l}kowski representation of $\Theta_1 - \Theta_2$.\footnote{Note that $\norm{X}_{\infty}$ can be expressed as $\inf \{\mu \in \mathbb{R} : X \leq \mu \1 \}$.} Since the Choi-Jamio{\l}kowski representation is linear we obtain
\begin{align}
\begin{split}
\varepsilon_{\Phi}= &\inf\limits_{\Xi}  \,  \norm{\Phi^{\setC}-\Xi\circ \Phi}_{\diamond} \\
& \st \, \Xi:\mathcal{S}(\mathcal{H}_B) \to \mathcal{S}(\mathcal{H}_E) \textnormal{ is cptp}
\end{split}
\end{align}
By using~\eqref{eq_wat} this can be rewritten as
\begin{align}
\begin{split}\label{eq:sdp1}
\varepsilon_{\Phi} = &2 \inf\limits_{\Xi} \inf\limits_{Z} \, \norm{\Trp{E}{Z}}_{\infty} \\
& \st \, Z \geq J(\Phi^{\setC})-J(\Xi\circ \Phi) \\
& \hspace{7.5mm} Z \geq 0 \\
& \hspace{7.5mm} J(\Xi) \geq 0 \\
& \hspace{7.5mm}  \Trp{E}{J(\Xi)} = \mathds{1}_B \,  \\
\end{split}
\end{align}
where the  two final constraints in \eqref{eq:sdp1}, i.e., $J(\Xi) \geq 0$ and $\Trp{E}{J(\Xi)} = \mathds{1}_B$ ensure that $\Xi$ is completely positive and trace-preserving. Since two minimizations can be always interchanged we can reformulate the optimization problem such that we obtain
\begin{align}
\begin{split}
	\varepsilon_{\Phi}=&2 \inf\limits_{Z,J(\Xi)} 	 \,	 \norm{\Trp{E}{Z}}_{\infty} \\
			&\st \, Z \geq J(\Phi^{\setC})-J(\Xi\circ \Phi) \\
			&  \hspace{7.5mm} Z \geq 0\\
			&  \hspace{7.5mm} J(\Xi) \geq 0\\
			&  \hspace{7.5mm} \Trp{E}{J(\Xi)} = \mathds{1}_B \, .						
\end{split} \label{eq:SSDDPP} 
\end{align}
It is now easy to see that~\eqref{eq:SSDDPP} is a semidefinite program. Note that for any cptp map $\Phi: \mathcal{S}(A) \to \mathcal{S}(B)$ we can reshuffle the Choi-Jamio{\l}kowski operator $J(\Phi)$ to a \emph{transfer matrix} $T(\Phi)$ defined by the involution $\left \langle i j \right| T(\Phi) \left| k \ell \right \rangle = \left \langle i k \right| J(\Phi) \left| j \ell \right \rangle$. Concatenating channels can be reduced to multiplying transfer matrices, i.e., the transfer matrix corresponding to the channel $\Xi \circ \Phi$ can be written as $T(\Xi \circ \Phi) = T(\Xi)\, T(\Phi)$ \cite{wolf_book}. Expressing~\eqref{eq:SSDDPP} in terms of transfer matrices thus shows that it is a semidefinite program.

\end{proof}

Semidefinite programs (SDPs) can be solved efficiently, i.e., in time that is polynomial in the program description size \cite{Khachiyan80}. We note that nowadays there exist several different algorithms that in practice solve SDPs very efficiently.
A good overview can be found, e.g., in \cite{ref:BoyVan-04,schrijver_book}. Therefore, for an arbitrary channel $\Phi$ its parameter $\varepsilon_{\Phi}$ (given in \eqref{opt:epsi}) that defines how close it is to being degradable can be evaluated efficiently. 

The conceptual idea we used above to derive upper bounds on the quantum and the private classical capacity is that a channel that is close to being degradable should have a channel coherent and channel private information that is nearly additive. The same idea can be applied to approximate anti-degradable channels.

\begin{mydef} [$\varepsilon$-anti-degradable] \label{def:epsiAntiDeg}
A channel $\Phi:\mathcal{S}(A) \to \mathcal{S}(B)$ is said to be \emph{$\varepsilon$-anti-degradable} if there exists a channel $\Xi:\mathcal{S}(E) \to \mathcal{S}(B)$ such that $\norm{\Phi-\Xi \circ \Phi^{\setC}}_{\diamond} \leq \varepsilon$.
\end{mydef}
 
\begin{mythm}[Properties of $\varepsilon$-anti-degradable channels] \label{thm:antideg}
If $\Phi:\mathcal{S}(A) \to \mathcal{S}(B)$ is an $\varepsilon$-anti-degradable channel, then 
\begin{align*}
Q(\Phi) &\leq P(\Phi) \\
 &\leq \frac{\varepsilon}{2} \log (|B|-1) + \varepsilon \log |B| + \Hb\Bigl(\frac{\varepsilon}{2}\Bigr) + \Bigl(1+\frac{\varepsilon}{2}\Bigr) \Hb\Bigl(\frac{\varepsilon}{2+\varepsilon}\Bigr) \, .
\end{align*}
\end{mythm}
\begin{proof}
The inequality $Q(\Phi) \leq P(\Phi)$ is straightforward since full quantum communication is necessarily private \cite[Thm.~12.6.3]{wilde_book}. Consider a cq state $\phi_{X A'_1 \ldots A'_n}$ and for $i \in [n]$ let  $V^i_{A'_i \to B_i E_i}$ denote the isometric extension of the $i$-th channel $\Phi$ and $U^i_{A'_i \to B_i E_i}$ denote the isometric extension of the $i$-th channel $\Xi \circ \Phi^{\setC}$. Let $\rho_{X B_1 E_1 \ldots B_n E_n}:=(\bigotimes_{i=1}^n V^i) \phi (\bigotimes_{i=1}^n {V^i}^{\dagger})$ and $\sigma^{(i)}:=(U^1 \otimes \ldots \otimes U^i \otimes V^{i+1} \otimes \ldots \otimes V^n) \, \phi \,	({U^1}^\dagger \otimes \ldots \otimes {U^i}^\dagger \otimes {V^{i+1}}^\dagger \otimes \ldots \otimes {V^n}^\dagger)$. Suppose $\rho_{X B_1 E_1 \ldots B_n E_n}$ is the state that maximizes $P^{(1)}(\Phi^{\otimes n})$ and let $\xi:= \frac{\varepsilon}{2} \log (|B|-1) + \varepsilon \log(|B|) + \Hb(\varepsilon/2) + (1+\frac{\varepsilon}{2}) \Hb(\frac{\varepsilon}{2+\varepsilon})$, then
\begin{align}
P^{(1)}(\Phi^{\otimes n})
&\hspace{0mm}= \I{X}{B_1 \ldots B_n}_{\rho} - \I{X}{E_1 \ldots E_n}_{\rho} \nonumber \\
&\hspace{0mm}\leq \I{X}{B_1 \ldots B_n}_{\sigma^{(1)}} - \I{X}{E_1 \ldots E_n}_{\sigma^{(1)}} + \xi\nonumber\\
&\hspace{0mm}\leq \I{X}{B_1 \ldots B_n}_{\sigma^{(n)}} - \I{X}{E_1 \ldots E_n}_{\sigma^{(n)}} + n\xi \label{eq:seq}\\
&\hspace{0mm}\leq n\xi \, ,\nonumber
\end{align}
where the first inequality follows by the strengthened Alicki-Fannes inequality (see Lemma~\ref{lem:alickiFannes}) together with the Fannes-Audenaert inequality (see Lemma~\ref{lem:fannesAudenaert}) and the fact that $\norm{\rho - \sigma^{(1)}}_{\trnorm} \leq \varepsilon$ which follows by the assumption $\norm{\Phi-\Xi \circ \Phi^{\setC}}_{\diamond}$. The inequality \eqref{eq:seq} follows by applying the same argument in sequence for all $1\leq i \leq n$. Note that by assumption $\norm{\sigma^{(i)}-\sigma^{(i+1)}}_{\trnorm} \leq \varepsilon$ for all $i \in [n-1]$. The final inequality uses that by construction the state $\sigma^{(n)}$ is generated by sending $\phi$ through $n$ copies of an anti-degradable channel. Anti-degradable channels are known to have a private capacity that is zero \cite{smith08_3}.
\end{proof}

Similar as for $\varepsilon$-degradable channels, given a channel $\Phi:\mathcal{S}(A) \to \mathcal{S}(B)$ we can consider 
\begin{align} \label{opt:epsiAD}
\bar\varepsilon_{\Phi}:=\left \lbrace
\begin{array}{lll}
			&\inf\limits_{\Xi} 		& \norm{\Phi-\Xi\circ \Phi^{\setC}}_{\diamond} \\
			&\st					& \Xi:\mathcal{S}(B) \to \mathcal{S}(E) \textnormal{ is cptp} \, ,
	\end{array} \right.
\end{align}
which defines the smallest possible parameter $\varepsilon$ such that the channel $\Phi$ is $\varepsilon$-anti-degradable.
Since $(\Phi^\setC)^\setC = \Phi$ we have $\bar\varepsilon_{\Phi} = \varepsilon_{\Phi^\setC}$. Proposition~\ref{prop:SDP} thus implies that \eqref{opt:epsiAD} can also be phrased as an SDP.
Theorem~\ref{thm:antideg} of course is valid for $\bar \varepsilon_{\Phi}$.

\vspace{4mm}

\section{Upper bounds via convex decompositions of channels} \label{sec:convexUB}

In this section we show how to combine the concept of $\varepsilon$-degradable channels with a standard technique to derive upper bounds to the quantum capacity that is based on the idea of decomposing an arbitrary channel $\Phi$ into a convex sum of approximate degradable channels. 

It has been shown that for an arbitrary quantum channel $\Phi$ the mapping $\Phi \mapsto Q(\Phi)$ is convex if $\Phi$ is \mbox{(anti-)}degradable \cite{wolf07}. Therefore, if a channel $\Phi$ can be written as a convex combination of \mbox{(anti-)}degradable channels, i.e., $\Phi = \sum_{i=1}^n p_i \Xi_i$, where $p \in \Delta_n$ and $\{\Xi_i \}_{i=1}^n$ are \mbox{(anti-)}degradable, $Q(\Phi) \leq \sum_{i=1}^n p_i Q(\Xi_i) = \sum_{i=1}^n p_i Q^{(1)}(\Xi_i)$ which describes a single-letter upper bound to the quantum capacity of $\Phi$ that can be powerful as demonstrated in \cite{smith08,smith08_2}. A drawback of this technique is that it is channel specific, i.e., the convex decomposition into degradable channels has to be reconstructed from scratch for every different channel. In addition, for an arbitrary channel, it is unclear how to efficiently find a convex decomposition of degradable channels --- even worse it is highly questionable if this is even possible in general. The extreme points of the set of all qubit channels have been shown to be degradable or anti-degradable channels \cite{wolf07,cubitt08} and therefore for qubit channels a convex decomposition into \mbox{(anti-)}degradable channels does exist, even if it might be difficult to find. However, a characterization of the extreme points of quantum channels with an input dimension larger than two is unknown \cite{ruskai07} and as such there is no reason to believe that an arbitrary quantum channel can be written as a convex combination of \mbox{(anti-)}degradable channels.

Recall the definitions of the symmetric side-channel assisted quantum and private classical capacities~\cite{smith08_2,smith08_3}:
\begin{align}
  Q_{ss}(\Phi) &:= \sup_{\Theta} Q(\Phi\otimes\Theta) = \sup_{\Theta} Q^{(1)}(\Phi\otimes\Theta) \label{eq_Qss} \\
  P_{ss}(\Phi) &:= \sup_{\Theta} P(\Phi\otimes\Theta) = \sup_{\Theta} P^{(1)}(\Phi\otimes\Theta) \, ,
\end{align}
where $\Theta$ ranges over so-called symmetric channels, i.e.~those with $\Theta = \Theta^\setC$. By definition, $Q(\Phi) \leq Q_{ss}(\Phi)$ and $P(\Phi) \leq  P_{ss}(\Phi)$. 
The important insights of~\cite{smith08_2,smith08_3} were that both $Q_{ss}$ and $P_{ss}$ have single-letter formulas, respectively; both are convex in the channel; and both coincide with $Q^{(1)}$ for degradable channels. The following theorem extends this insight to approximate degradable channels.

\vspace{2mm}

\begin{mythm}
  \label{thm:Qss-vs-U}
  If $\Phi$ is an $\varepsilon$-degradable channel, with a degrading channel $\Xi$, then
\begin{enumerate}[(i)]
\item $Q_{ss}(\Phi) \leq U_{\Xi}(\Phi) + \varepsilon\log |E| + ( 1+\frac{\varepsilon}{2} ) \Hb(\frac{\varepsilon}{2+\varepsilon})$ , \label{itt_i}
\item $P_{ss}(\Phi) \leq U_{\Xi}(\Phi) + \varepsilon (2\log |E|+\frac{1}{2}\log|F| ) + \frac{5}{2} (1+\frac{\varepsilon}{2} ) \Hb(\frac{\varepsilon}{2+\varepsilon})$ . \label{itt_ii}
\end{enumerate}
\end{mythm}
We note that the statement~\eqref{itt_i} of Theorem~\ref{thm:Qss-vs-U} implies the statement~\eqref{it:i2} of Theorem~\ref{thm:epsilonDegradableAdditivity}.
\begin{proof}
Optimizing over symmetric side channels with generic 
isometry $U:A' \hookrightarrow B'\otimes E'$, and input states $\rho_{AA'}$
(or rather its purification), we have (using the same notation as in
the proof of Theorem~\ref{thm:epsilonDegradableAdditivity} otherwise for the approximate degrading channel): 
\begin{align}
  Q_{ss}(\Phi) &= \sup I(AA' \rangle BB') \, ,
\end{align}
where the supremum is over all input states $\rho_{AA'}$ and over all symmetric side channels (see~\eqref{eq_Qss}). With the same notation we then find
\begin{align}
  Q_{ss}(\Phi) 
                      &\hspace{0mm}= \sup H(BB') - H(EE') \\
                      &\hspace{0mm}= \sup H(B|E') - H(E|E') \\
                      &\hspace{0mm}\leq \sup H\bigl(F\widetilde{E}|E'\bigr) - H\bigl(\widetilde{E}|E'\bigr) + \varepsilon\log |E| + \left( 1+\frac{\varepsilon}{2} \right) h\!\left(\frac{\varepsilon}{2+\varepsilon}\right) \\
                      &\hspace{0mm}= \sup H\bigl(F|\widetilde{E}E'\bigr) + \varepsilon\log |E|  + \left( 1+\frac{\varepsilon}{2} \right) h\!\left(\frac{\varepsilon}{2+\varepsilon}\right) \\
                      &\hspace{0mm}= \sup H\bigl(F|\widetilde{E}\bigr) + \varepsilon\log |E| + \left( 1+\frac{\varepsilon}{2} \right) h\!\left(\frac{\varepsilon}{2+\varepsilon}\right).
\end{align}
Here, the first two lines are by definition, the third by the symmetry between $B'$ and
$E'$; in the fourth we use the Stinespring isometry $W$ of the approximate
degrading channel and Lemma~\ref{lem:alickiFannes}; in the fifth line we rewrite the
difference of conditional entropies using the chain rule, and in the last
line we have ``$\leq$'' by strong subadditivity, but equality is achieved with a
trivial symmetric side channel. But now, $\sup H\bigl(F|\widetilde{E}\bigr)$ contains
only the maximization over input states $\rho_A$, giving $U_{\Xi}(\Phi)$.

The proof for $P_{ss}$ is similar, cf.~Theorem~\ref{thm:epsilonDegradableAdditivity}.
\end{proof}

The significance of Theorem~\ref{thm:Qss-vs-U} is that $Q_{ss}$ is convex, unlike $Q$ (and likewise $P_{ss}$, in contrast to $P$). As a result, we even get strengthened upper bounds by taking the convex hull of the bound in Theorem~\ref{thm:epsilonDegradableAdditivity} and other upper bounds on $Q_{ss}$~\cite{smith08_2,ouyang14}. This is done in Section~\ref{sec:examples} to derive upper bounds for the capacity of a depolarizing and a BB84 channel.

\section{Applications} \label{sec:examples}
We now illustrate the power of the bounds derived in the previous sections on three examples. Recall that the upper bounds for the quantum and the private classical capacity derived in Sections~\ref{sec:main} and \ref{sec:convexUB} are valid for arbitrary channels $\Phi$. As the parameter $\varepsilon_{\Phi}$ given in \eqref{opt:epsi} is described via an SDP, it can be evaluated efficiently for every possible channel. Thus our upper bounds can be immediately applied and efficiently evaluated for arbitrary channels, whereas most previous upper bounds rely on channel specific constructions which can be different for each channel and are usually difficult to find \cite{smith08_2,smith08}. As we will see in this section, we can also combine different upper bounds, i.e., taking the best known upper bounds for every scenario.

\subsection{Depolarizing channel} \label{ex:depolarizingChannel}
Consider a depolarizing channel $\mathcal{D}_p: \mathcal{S}(A) \ni \rho \mapsto (1-p)\rho + \tfrac{p}{3}(X\rho X+Z\rho Z + Y \rho Y) \in \mathcal{S}(B)$ with $\dim A = \dim B =2$ and $p \in [0,1]$. Its channel coherent information is maximized on a Bell state as input, and hence evaluates to~\cite{bennett_mixed-state_1996,shorshort} (see also \cite[p.~575]{wilde_book})
\begin{equation}
Q^{(1)}(\mathcal{D}_p) = 1 + (1-p)\log(1-p) + p \log\left(\frac{p}{3} \right)\, . \label{eq:ciDepol}
\end{equation}
A well-known upper bound to $Q(\Phi)$ has been derived in~\cite[Cor.~7]{smith08} (see also~\cite[Thm.~5.5]{ouyang14}) and is given for $0 \leq p \leq \frac{1}{4}$ and $\gamma(p) := 4 (\sqrt{1-p} -1 +p)$ by
\begin{align}
&Q(\mathcal{D}_p) \leq \mathrm{conv}\Big \{ 1-\Hb(p),\Hb\left(\frac{1+\gamma(p)}{2}\Bigr) - \Hb\Bigl(\frac{\gamma(p)}{2} \right),  1-4p\Big \} \ . \label{eq_JapaneseUB}
\end{align}
Combining Theorem~\ref{thm:epsilonDegradableAdditivity} and Theorem~\ref{thm:Qss-vs-U} we obtain the upper bound
\begin{align} 
&Q(\mathcal{D}_p) \leq \mathrm{conv}\Big \{  U_{\Xi}(\mathcal{D}_p), 1-\Hb(p), \Hb\left(\frac{1+\gamma(p)}{2}\right) - \Hb\left(\frac{\gamma(p)}{2} \right), 1-4p\Big \} \ .\label{eq_newUB_moreT} 
\end{align}
We recall that $U_{\Xi}(\mathcal{D}_p)$ is given via a convex optimization problem. In order to further simplify our upper bound, using the structure of the depolarizing channel we further bound $U_{\Xi}(\mathcal{D}_p)$ with the help of Proposition~\ref{prop:Q1-vs-U}.
\begin{align} 
&Q(\mathcal{D}_p) \leq \mathrm{conv}\Big \{ Q^{(1)}(\Phi) +  \frac{\varepsilon_{p}}{2}  \log(|E|-1) +  \Hb\left(\frac{\varepsilon_{p}}{2}\right)   +\varepsilon_{p} \log |E|  + \left(1+\frac{\varepsilon_{p}}{2} \right) \Hb\left(\frac{\varepsilon_{p}}{2+\varepsilon_{p}} \right),  \nonumber \\
&\hspace{30mm} 1-\Hb(p),\Hb\left(\frac{1+\gamma(p)}{2}\right) - \Hb\left(\frac{\gamma(p)}{2} \right), 1-4p\Big \} \ .\label{eq_newUB}
\end{align}
Figure~\ref{fig:depol} compares the new upper bound given by~\eqref{eq_newUB}, for $\varepsilon_{p}$ as in \eqref{opt:epsi}, with the upper bound given by~\eqref{eq_JapaneseUB}. We note that the upper bound~\eqref{eq_newUB_moreT} can be potentially considerably better than~\eqref{eq_newUB}, however one has to solve a convex optimization problem that defines $U_{\Xi}(\mathcal{D}_p)$.\footnote{Very recently after this paper, Leditzky \emph{et al.}~presented a new upper bound for the depolarizing channel~\cite{letditzky17} that outperforms~\eqref{eq_newUB} in the high noise regime. See~\cite{letditzky17} for a comparison of the new bound with~\eqref{eq_newUB}.}

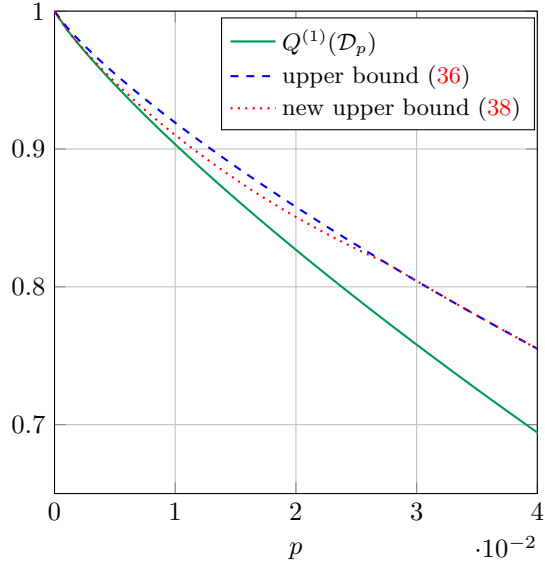
\begin{figure}[!htb]
\centering

  \begin{tikzpicture}
	\begin{axis}[
		height=8.0cm,
		width=8.0cm,
		grid=major,
		xlabel=$p$,
		xmin=0,
		xmax=0.04,
		ymax=1.0,
		ymin=0.65,
	     xtick={0,0.01,0.02,0.03,0.04},
          ytick={1,0.9,0.8,0.7,0.6},
		legend style={at={(0.67,0.987)},anchor=north,legend cell align=left,font=\footnotesize} 
	]



	\addplot[ForestGreen,thick,smooth] coordinates {
(0.,1) (0.001,0.987007) (0.002,0.976016) (0.003,0.965781) (0.004,0.956038) (0.005,0.94666) (0.006,0.937575) (0.007,0.928733) (0.008,0.920099) (0.009,0.911647) (0.01,0.903357) (0.011,0.895213) (0.012,0.887203) (0.013,0.879313) (0.014,0.871537) (0.015,0.863865) (0.016,0.856291) (0.017,0.848808) (0.018,0.841412) (0.019,0.834097) (0.02,0.82686) (0.021,0.819697) (0.022,0.812603) (0.023,0.805577) (0.024,0.798615) (0.025,0.791715) (0.026,0.784874) (0.027,0.77809) (0.028,0.77136) (0.029,0.764684) (0.03,0.758059) (0.031,0.751484) (0.032,0.744957) (0.033,0.738476) (0.034,0.73204) (0.035,0.725649) (0.036,0.7193) (0.037,0.712992) (0.038,0.706726) (0.039,0.700498) (0.04,0.694309) (0.041,0.688158) (0.042,0.682043) (0.043,0.675965) (0.044,0.669921) (0.045,0.663912) (0.046,0.657936) (0.047,0.651993) (0.048,0.646082) (0.049,0.640203) (0.05,0.634355) (0.051,0.628537) (0.052,0.622749) (0.053,0.61699) (0.054,0.61126) (0.055,0.605559) (0.056,0.599885) (0.057,0.594238) (0.058,0.588618) (0.059,0.583025) (0.06,0.577457) (0.061,0.571915) (0.062,0.566399) (0.063,0.560907) (0.064,0.555439) (0.065,0.549996) (0.066,0.544577) (0.067,0.53918) (0.068,0.533807) (0.069,0.528457) (0.07,0.523129) (0.071,0.517823) (0.072,0.512539) (0.073,0.507277) (0.074,0.502036) (0.075,0.496816) (0.076,0.491617) (0.077,0.486439) (0.078,0.48128) (0.079,0.476142) (0.08,0.471024) (0.081,0.465925) (0.082,0.460846) (0.083,0.455785) (0.084,0.450744) (0.085,0.445722) (0.086,0.440718) (0.087,0.435732) (0.088,0.430765) (0.089,0.425815) (0.09,0.420884) (0.091,0.41597) (0.092,0.411073) (0.093,0.406194) (0.094,0.401331) (0.095,0.396486) (0.096,0.391657) (0.097,0.386846) (0.098,0.38205) (0.099,0.377271) (0.1,0.372508)
	};
	\addlegendentry{$Q^{(1)}(\mathcal{D}_p)$}

%
		
		\addplot[dashed,blue,thick,smooth] coordinates {
(0.,1) (0.001,0.988592) (0.002,0.979183) (0.003,0.970529) (0.004,0.962363) (0.005,0.954561) (0.006,0.947048) (0.007,0.939775) (0.008,0.932706) (0.009,0.925817) (0.01,0.919086) (0.011,0.912498) (0.012,0.90604) (0.013,0.899699) (0.014,0.893468) (0.015,0.887337) (0.016,0.881301) (0.017,0.875352) (0.018,0.869486) (0.019,0.863698) (0.02,0.857984) (0.021,0.85234) (0.022,0.846761) (0.023,0.841246) (0.024,0.835792) (0.025,0.830395) (0.026,0.825053) (0.027,0.819764) (0.028,0.814526) (0.029,0.809338) (0.03,0.804196) (0.031,0.7991) (0.032,0.794048) (0.033,0.789038) (0.034,0.78407) (0.035,0.779141) (0.036,0.774251) (0.037,0.769399) (0.038,0.764582) (0.039,0.759802) (0.04,0.755055) (0.041,0.750342) (0.042,0.745661) (0.043,0.741012) (0.044,0.736394) (0.045,0.731806) (0.046,0.727247) (0.047,0.722717) (0.048,0.718215) (0.049,0.71374) (0.05,0.709292) (0.051,0.70487) (0.052,0.700473) (0.053,0.696102) (0.054,0.691754) (0.055,0.687431) (0.056,0.683132) (0.057,0.678855) (0.058,0.674601) (0.059,0.670369) (0.06,0.666158) (0.061,0.661969) (0.062,0.6578) (0.063,0.653653) (0.064,0.649525) (0.065,0.645417) (0.066,0.641328) (0.067,0.637258) (0.068,0.633207) (0.069,0.629174) (0.07,0.62516) (0.071,0.621163) (0.072,0.617184) (0.073,0.613222) (0.074,0.609276) (0.075,0.605348) (0.076,0.601435) (0.077,0.597539) (0.078,0.593659) (0.079,0.589795) (0.08,0.585945) (0.081,0.582112) (0.082,0.578292) (0.083,0.574488) (0.084,0.570698) (0.085,0.566923) (0.086,0.563161) (0.087,0.559414) (0.088,0.55568) (0.089,0.55196) (0.09,0.548253) (0.091,0.544559) (0.092,0.540878) (0.093,0.53721) (0.094,0.533555) (0.095,0.529912) (0.096,0.526281) (0.097,0.522663) (0.098,0.519056) (0.099,0.515461) (0.1,0.511878)			
	};
		\addlegendentry{upper bound~\eqref{eq_JapaneseUB}}


		\addplot[dotted,red,thick,smooth] coordinates {
    (0.0000,    1.0000)
    (0.0010,    0.9870)
    (0.0020,   0.9763 )
    (0.0030,    0.9665)
    (0.0040,   0.9573 )
    (0.0050,   0.9485 )
    (0.0060,    0.9402)
    (0.0070,    0.9322)
    (0.0080,    0.9245)
    (0.0090,    0.9172)
    (0.0100,    0.9101)
    (0.0110,    0.9032)
    (0.0120,   0.8966 )
    (0.0130,    0.8902)
    (0.0140,    0.8840)
    (0.0150,   0.8780 )
    (0.0160,   0.8721 )
    (0.0170,    0.8665)
    (0.0180,   0.8611 )
    (0.0190,   0.8558 )
    (0.0200,   0.8507 )
    (0.0210,   0.8458 )
    (0.0220,    0.8410)
    (0.0230,    0.8363)
    (0.0240,    0.8319)
    (0.0250,    0.8275)
    (0.0260,    0.8233)
    (0.0270,    0.8193)
    (0.0280,    0.8145)
    (0.0290,    0.8093)
    (0.0300,    0.8042)
    (0.0310,    0.7991)
    (0.0320,   0.7940 )
    (0.0330,    0.7890)
    (0.0340,   0.7841 )
    (0.0350,    0.7791)
    (0.0360,    0.7743)
    (0.0370,    0.7694)
    (0.0380,   0.7646 )
    (0.0390,    0.7598)
    (0.0400,   0.7551 )
	};
		\addlegendentry{new upper bound~\eqref{eq_newUB}}


	\end{axis}  

\end{tikzpicture}
\caption{This plot depicts upper and lower bounds for the quantum capacity $Q(\mathcal{D}_p)$ of a qubit depolarizing channel $\mathcal{D}_p$. The channel coherent information given in \eqref{eq:ciDepol} denotes a lower bound on $Q(\mathcal{D}_p)$ (cf.\ solid green curve). The dashed blue curve denotes an upper bound on $Q(\mathcal{D}_p)$ given in~\eqref{eq_JapaneseUB}. The dotted red line depicts the upper bound given by~\eqref{eq_newUB} with $\varepsilon_{p}$ as given in \eqref{opt:epsi}.}
\label{fig:depol}
\end{figure}

It can be shown that for two copies of an arbitrary quantum channel $\Phi$ we have $\varepsilon_{\Phi \otimes \Phi} \leq 2 \varepsilon_{\Phi}$. To see this, we assume without loss of generality that $\Xi$ denotes the optimizer in~\eqref{opt:epsi}. Since $\Xi \otimes \Xi$ is a feasible degrading channel, we obtain
\begin{align*}
\varepsilon_{\Phi \otimes \Phi} &\leq \norm{ \Phi^\setC \otimes \Phi^\setC - (\Xi \circ \Phi)\otimes(\Xi \circ \Phi)}_{\diamond} \nonumber \\
&= \norm{(\Phi^\setC - \Xi \circ \Phi)\otimes \Phi^\setC - (\Xi \circ \Phi)\otimes(\Xi \circ \Phi - \Phi^\setC)}_{\diamond} \nonumber \\
&\leq \norm{\Phi^\setC - \Xi \circ \Phi}_{\diamond} \norm{\Phi^\setC}_{\diamond} + \norm{\Xi \circ \Phi}_{\diamond} \norm{\Phi^\setC - \Xi \circ \Phi}_{\diamond} \nonumber \\
&\leq 2 \varepsilon_{\Phi} \ ,
\end{align*}
where we used that the diamond norm is multiplicative under the tensor product.
It could happen that $\varepsilon_{\Phi \otimes \Phi}$ is considerably smaller than $2 \varepsilon_{\Phi}$, i.e., the tensor channel $\Phi \otimes \Phi$ would be noticeably closer to being degradable than $\Phi$. Numerics for the depolarizing channel however show that this is not the case (we observe that $\varepsilon_{\Phi \otimes \Phi}$ is always close to $2 \varepsilon_{\Phi}$).

\subsection{BB84 channel}
Consider a qubit Pauli channel with independent bit flip and phase flip error probability where $p_X \in [0,\tfrac{1}{2}]$ denotes the bit flip and $p_Z\in [0,\tfrac{1}{2}]$ the phase flip probability. Due to its relevance for the BB84 protocol this channel is often called BB84 channel. More formally this is a channel $\mathcal{B}_{p_X,p_Z}:\mathcal{S}(A) \ni \rho \mapsto (1-p_X-p_Z+p_X p_Z)\rho +(p_X - p_X p_Z )X \rho X + (p_Z-p_Z p_X)Z \rho Z + p_X p_Z Y \rho Y \in \mathcal{S}(B)$ with $\dim A = \dim B=2$. It is immediate to verify that a Bell state maximizes the coherent information and therefore the channel coherent information of the BB84 channel is given by $Q^{(1)}(\mathcal{B}_{p_X,p_Z})=1-\Hb(p_X)-\Hb(p_Z)$. 
For the case where $p_X=p_Z=:p$ it has been shown that \cite{smith08}
\begin{equation}
Q(\mathcal{B}_{p,p}) \leq \Hb \Bigl( \frac{1}{2} - 2 p(1-p) \Bigr) - \Hb\bigl(2 p(1-p) \bigr)\, . \label{eq:UBsmith}
\end{equation}
Combining Theorem~\ref{thm:epsilonDegradableAdditivity} and Theorem~\ref{thm:Qss-vs-U} we obtain the upper bound
\begin{align} 
&Q(\mathcal{B}_{p,p}) \leq \mathrm{conv}\Big \{ Q^{(1)}(\mathcal{B}_{p,p}) +  \frac{\varepsilon_{p}}{2}  \log(|E|-1) +  \Hb\left(\frac{\varepsilon_{p}}{2}\right) +\varepsilon_{p} \log |E| + \left(1+\frac{\varepsilon_{p}}{2} \right) \Hb\left(\frac{\varepsilon_{p}}{2+\varepsilon_{p}} \right), \nonumber \\
&\hspace{30mm} \Hb \left( \frac{1}{2} - 2 p(1-p) \right) - \Hb\bigl(2 p(1-p) \bigr)\Big \} \ ,\label{eq_newUB_BB84}
\end{align}
which is strictly better than~\eqref{eq:UBsmith}.
Figure~\ref{fig:bb84} compares the upper bound of the quantum capacity derived in Theorem~\ref{thm:epsilonDegradableAdditivity}, for $\varepsilon_{p}$ as in \eqref{opt:epsi}, with previously known upper bounds. In the high-noise regime~\eqref{eq:UBsmith} outperforms~\eqref{eq_newUB_BB84}. In most quantum key distribution protocols the secret-key rate is given by the difference of a min-entropy term and a term that comes from the error correction step (which corresponds to the quantum capacity). Oftentimes these two terms are of the same magnitude and therefore the improvement of~\eqref{eq_newUB_BB84} compared to~\eqref{eq:UBsmith} which looks small on Figure~\ref{fig:bb84} can be very relevant for bounding the secret-key rate.     
We note that the upper bound~\eqref{eq_newUB_BB84} can be slightly improved by using the $U_{\Xi}(\mathcal{B}_{p,p})$ quantity.
\begin{figure}[!htb]
\centering
\hspace{-7mm}
    \subfloat[$p_X=p_Z=p$]{ 
  \begin{tikzpicture}
	\begin{axis}[
		height=8cm,
		width=8cm,
		grid=major,
		xlabel=$p$,
		xmin=0,
		xmax=0.01,
		ymax=1.0,
		ymin=0.8,
	     xtick={0,0.0025,0.005,0.0075,0.01},
          ytick={1,0.95,0.9,0.85,0.8},
		legend style={at={(0.72,0.99)},anchor=north,legend cell align=left,font=\footnotesize} 
	]



	\addplot[ForestGreen,thick,smooth] coordinates {
(0.,1) (0.0001,0.997054) (0.0002,0.994508) (0.0003,0.992113) (0.0004,0.989816) (0.0005,0.987592) (0.0006,0.985426) (0.0007,0.983308) (0.0008,0.981232) (0.0009,0.979192) (0.001,0.977184) (0.0011,0.975206) (0.0012,0.973253) (0.0013,0.971325) (0.0014,0.969418) (0.0015,0.967533) (0.0016,0.965666) (0.0017,0.963818) (0.0018,0.961987) (0.0019,0.960172) (0.002,0.958372) (0.0021,0.956586) (0.0022,0.954815) (0.0023,0.953056) (0.0024,0.95131) (0.0025,0.949576) (0.0026,0.947854) (0.0027,0.946143) (0.0028,0.944442) (0.0029,0.942752) (0.003,0.941072) (0.0031,0.939401) (0.0032,0.93774) (0.0033,0.936088) (0.0034,0.934445) (0.0035,0.93281) (0.0036,0.931183) (0.0037,0.929565) (0.0038,0.927954) (0.0039,0.926351) (0.004,0.924755) (0.0041,0.923167) (0.0042,0.921586) (0.0043,0.920011) (0.0044,0.918443) (0.0045,0.916882) (0.0046,0.915328) (0.0047,0.913779) (0.0048,0.912237) (0.0049,0.910701) (0.005,0.909171) (0.0051,0.907646) (0.0052,0.906127) (0.0053,0.904614) (0.0054,0.903107) (0.0055,0.901604) (0.0056,0.900107) (0.0057,0.898615) (0.0058,0.897128) (0.0059,0.895647) (0.006,0.89417) (0.0061,0.892698) (0.0062,0.891231) (0.0063,0.889768) (0.0064,0.88831) (0.0065,0.886857) (0.0066,0.885408) (0.0067,0.883963) (0.0068,0.882523) (0.0069,0.881087) (0.007,0.879655) (0.0071,0.878228) (0.0072,0.876804) (0.0073,0.875385) (0.0074,0.873969) (0.0075,0.872558) (0.0076,0.87115) (0.0077,0.869746) (0.0078,0.868346) (0.0079,0.86695) (0.008,0.865557) (0.0081,0.864168) (0.0082,0.862782) (0.0083,0.861401) (0.0084,0.860022) (0.0085,0.858647) (0.0086,0.857276) (0.0087,0.855908) (0.0088,0.854543) (0.0089,0.853181) (0.009,0.851823) (0.0091,0.850468) (0.0092,0.849117) (0.0093,0.847768) (0.0094,0.846422) (0.0095,0.84508) (0.0096,0.843741) (0.0097,0.842405) (0.0098,0.841071) (0.0099,0.839741) (0.01,0.838414) (0.0101,0.837089) (0.0102,0.835768) (0.0103,0.834449) (0.0104,0.833133) (0.0105,0.83182) (0.0106,0.83051) (0.0107,0.829202) (0.0108,0.827898) (0.0109,0.826596) (0.011,0.825296) (0.0111,0.823999) (0.0112,0.822705) (0.0113,0.821414) (0.0114,0.820125) (0.0115,0.818838) (0.0116,0.817555) (0.0117,0.816273) (0.0118,0.814994) (0.0119,0.813718) (0.012,0.812444) (0.0121,0.811173) (0.0122,0.809904) (0.0123,0.808637) (0.0124,0.807373) (0.0125,0.806111) (0.0126,0.804851) (0.0127,0.803594) (0.0128,0.802339) (0.0129,0.801086) (0.013,0.799836) (0.0131,0.798588) (0.0132,0.797342) (0.0133,0.796098) (0.0134,0.794856) (0.0135,0.793617) (0.0136,0.79238) (0.0137,0.791145) (0.0138,0.789912) (0.0139,0.788681) (0.014,0.787453) (0.0141,0.786226) (0.0142,0.785002) (0.0143,0.783779) (0.0144,0.782559) (0.0145,0.78134) (0.0146,0.780124) (0.0147,0.77891) (0.0148,0.777697) (0.0149,0.776487) (0.015,0.775279) (0.0151,0.774072) (0.0152,0.772868) (0.0153,0.771665) (0.0154,0.770464) (0.0155,0.769266) (0.0156,0.768069) (0.0157,0.766874) (0.0158,0.765681) (0.0159,0.764489) (0.016,0.7633) (0.0161,0.762112) (0.0162,0.760927) (0.0163,0.759743) (0.0164,0.758561) (0.0165,0.75738) (0.0166,0.756202) (0.0167,0.755025) (0.0168,0.75385) (0.0169,0.752676) (0.017,0.751505) (0.0171,0.750335) (0.0172,0.749167) (0.0173,0.748) (0.0174,0.746836) (0.0175,0.745673) (0.0176,0.744511) (0.0177,0.743351) (0.0178,0.742193) (0.0179,0.741037) (0.018,0.739882) (0.0181,0.738729) (0.0182,0.737578) (0.0183,0.736428) (0.0184,0.73528) (0.0185,0.734133) (0.0186,0.732988) (0.0187,0.731844) (0.0188,0.730702) (0.0189,0.729562) (0.019,0.728423) (0.0191,0.727286) (0.0192,0.72615) (0.0193,0.725016) (0.0194,0.723883) (0.0195,0.722752) (0.0196,0.721623) (0.0197,0.720494) (0.0198,0.719368) (0.0199,0.718243) (0.02,0.717119)
	};
	\addlegendentry{$Q^{(1)}(\mathcal{B}_{p,p})$}

		\addplot[dashed,blue,thick,smooth] coordinates {
(0.,1) (0.0001,0.997254) (0.0002,0.994908) (0.0003,0.992714) (0.0004,0.990618) (0.0005,0.988594) (0.0006,0.986629) (0.0007,0.984713) (0.0008,0.982838) (0.0009,0.980999) (0.001,0.979192) (0.0011,0.977415) (0.0012,0.975664) (0.0013,0.973937) (0.0014,0.972232) (0.0015,0.970548) (0.0016,0.968883) (0.0017,0.967236) (0.0018,0.965607) (0.0019,0.963994) (0.002,0.962395) (0.0021,0.960812) (0.0022,0.959242) (0.0023,0.957685) (0.0024,0.956141) (0.0025,0.954609) (0.0026,0.953089) (0.0027,0.951579) (0.0028,0.950081) (0.0029,0.948593) (0.003,0.947114) (0.0031,0.945646) (0.0032,0.944187) (0.0033,0.942737) (0.0034,0.941295) (0.0035,0.939862) (0.0036,0.938438) (0.0037,0.937021) (0.0038,0.935613) (0.0039,0.934212) (0.004,0.932818) (0.0041,0.931431) (0.0042,0.930052) (0.0043,0.92868) (0.0044,0.927314) (0.0045,0.925955) (0.0046,0.924602) (0.0047,0.923256) (0.0048,0.921915) (0.0049,0.920581) (0.005,0.919253) (0.0051,0.91793) (0.0052,0.916613) (0.0053,0.915302) (0.0054,0.913996) (0.0055,0.912696) (0.0056,0.9114) (0.0057,0.91011) (0.0058,0.908825) (0.0059,0.907545) (0.006,0.90627) (0.0061,0.905) (0.0062,0.903734) (0.0063,0.902473) (0.0064,0.901217) (0.0065,0.899965) (0.0066,0.898717) (0.0067,0.897474) (0.0068,0.896236) (0.0069,0.895001) (0.007,0.893771) (0.0071,0.892544) (0.0072,0.891322) (0.0073,0.890104) (0.0074,0.88889) (0.0075,0.887679) (0.0076,0.886473) (0.0077,0.88527) (0.0078,0.884071) (0.0079,0.882876) (0.008,0.881684) (0.0081,0.880496) (0.0082,0.879312) (0.0083,0.878131) (0.0084,0.876953) (0.0085,0.875779) (0.0086,0.874608) (0.0087,0.873441) (0.0088,0.872277) (0.0089,0.871116) (0.009,0.869959) (0.0091,0.868804) (0.0092,0.867653) (0.0093,0.866505) (0.0094,0.86536) (0.0095,0.864218) (0.0096,0.863079) (0.0097,0.861943) (0.0098,0.86081) (0.0099,0.85968) (0.01,0.858552) (0.0101,0.857428) (0.0102,0.856307) (0.0103,0.855188) (0.0104,0.854072) (0.0105,0.852959) (0.0106,0.851848) (0.0107,0.85074) (0.0108,0.849635) (0.0109,0.848533) (0.011,0.847433) (0.0111,0.846336) (0.0112,0.845241) (0.0113,0.844149) (0.0114,0.843059) (0.0115,0.841972) (0.0116,0.840888) (0.0117,0.839806) (0.0118,0.838726) (0.0119,0.837649) (0.012,0.836574) (0.0121,0.835501) (0.0122,0.834431) (0.0123,0.833363) (0.0124,0.832298) (0.0125,0.831234) (0.0126,0.830173) (0.0127,0.829115) (0.0128,0.828058) (0.0129,0.827004) (0.013,0.825952) (0.0131,0.824902) (0.0132,0.823854) (0.0133,0.822809) (0.0134,0.821766) (0.0135,0.820724) (0.0136,0.819685) (0.0137,0.818648) (0.0138,0.817613) (0.0139,0.81658) (0.014,0.815549) (0.0141,0.81452) (0.0142,0.813493) (0.0143,0.812469) (0.0144,0.811446) (0.0145,0.810425) (0.0146,0.809406) (0.0147,0.808389) (0.0148,0.807373) (0.0149,0.80636) (0.015,0.805349) (0.0151,0.804339) (0.0152,0.803332) (0.0153,0.802326) (0.0154,0.801322) (0.0155,0.80032) (0.0156,0.79932) (0.0157,0.798322) (0.0158,0.797325) (0.0159,0.79633) (0.016,0.795337) (0.0161,0.794346) (0.0162,0.793356) (0.0163,0.792369) (0.0164,0.791382) (0.0165,0.790398) (0.0166,0.789415) (0.0167,0.788434) (0.0168,0.787455) (0.0169,0.786478) (0.017,0.785502) (0.0171,0.784527) (0.0172,0.783555) (0.0173,0.782584) (0.0174,0.781614) (0.0175,0.780647) (0.0176,0.77968) (0.0177,0.778716) (0.0178,0.777753) (0.0179,0.776791) (0.018,0.775831) (0.0181,0.774873) (0.0182,0.773916) (0.0183,0.772961) (0.0184,0.772008) (0.0185,0.771055) (0.0186,0.770105) (0.0187,0.769156) (0.0188,0.768208) (0.0189,0.767262) (0.019,0.766317) (0.0191,0.765374) (0.0192,0.764432) (0.0193,0.763492) (0.0194,0.762553) (0.0195,0.761615) (0.0196,0.760679) (0.0197,0.759745) (0.0198,0.758812) (0.0199,0.75788) (0.02,0.75695)
	};
		\addlegendentry{UB \eqref{eq:UBsmith} \cite{smith08}}	
		
		\addplot[dotted,red,thick,smooth] coordinates {

(0,1)
(0.0001,0.99705)
(0.0002,0.99451)
(0.0003,0.99211)
(0.0004,0.98982)
(0.0005,0.98759)
(0.0006,0.98543)
(0.0007,0.98331)
(0.0008,0.98126)
(0.0009,0.97947)
(0.001,0.97752)
(0.0011,0.97561)
(0.0012,0.97373)
(0.0013,0.97188)
(0.0014,0.97005)
(0.0015,0.96825)
(0.0016,0.96648)
(0.0017,0.96473)
(0.0018,0.963)
(0.0019,0.96129)
(0.002,0.9596)
(0.0021,0.95793)
(0.0022,0.95628)
(0.0023,0.95465)
(0.0024,0.95304)
(0.0025,0.95144)
(0.0026,0.94985)
(0.0027,0.94829)
(0.0028,0.94674)
(0.0029,0.9452)
(0.003,0.94368)
(0.0031,0.94217)
(0.0032,0.94068)
(0.0033,0.93919)
(0.0034,0.93773)
(0.0035,0.93627)
(0.0036,0.93483)
(0.0037,0.9334)
(0.0038,0.93198)
(0.0039,0.93057)
(0.004,0.92918)
(0.0041,0.9278)
(0.0042,0.92643)
(0.0043,0.92507)
(0.0044,0.92372)
(0.0045,0.92238)
(0.0046,0.92105)
(0.0047,0.91973)
(0.0048,0.91842)
(0.0049,0.91713)
(0.005,0.91584)
(0.0051,0.91456)
(0.0052,0.91329)
(0.0053,0.91203)
(0.0054,0.91078)
(0.0055,0.90954)
(0.0056,0.90831)
(0.0057,0.90709)
(0.0058,0.90588)
(0.0059,0.90468)
(0.006,0.90348)
(0.0061,0.9023)
(0.0062,0.90112)
(0.0063,0.89995)
(0.0064,0.89879)
(0.0065,0.89764)
(0.0066,0.89649)
(0.0067,0.89536)
(0.0068,0.89423)
(0.0069,0.89311)
(0.007,0.892)
(0.0071,0.89089)
(0.0072,0.8898)
(0.0073,0.88871)
(0.0074,0.88763)
(0.0075,0.88656)
(0.0076,0.88549)
(0.0077,0.88443)
(0.0078,0.88338)
(0.0079,0.88234)
(0.008,0.8813)
(0.0081,0.88028)
(0.0082,0.87925)
(0.0083,0.87813)
(0.0084,0.87695)
(0.0085,0.87578)
(0.0086,0.87461)
(0.0087,0.87344)
(0.0088,0.87228)
(0.0089,0.87112)
(0.009,0.86996)
(0.0091,0.8688)
(0.0092,0.86765)
(0.0093,0.8665)
(0.0094,0.86536)
(0.0095,0.86422)
(0.0096,0.86308)
(0.0097,0.86194)
(0.0098,0.86081)
(0.0099,0.85968)
(0.01,0.85855)
	};
		\addlegendentry{new UB~\eqref{eq_newUB_BB84}}			
		

	\end{axis}  

\end{tikzpicture} }  \qquad \qquad
    \subfloat[$p_Z=100 p_X$]{
  \begin{tikzpicture}
	\begin{axis}[
		height=8cm,
		width=8cm,
		grid=major,
		xlabel=$p_X$,
		xmin=0,
		xmax=0.001,
		ymax=1.0,
		ymin=0.5,
	     xtick={0,0.00025,0.0005,0.00075,0.001},
          ytick={1,0.9,0.8,0.7,0.6,0.5},
		legend style={at={(0.627,0.99)},anchor=north,legend cell align=left,font=\footnotesize} 
	]



	\addplot[ForestGreen,thick,smooth] coordinates {
(0.,1) (0.00001,0.988412) (0.00002,0.978845) (0.00003,0.970042) (0.00004,0.961736) (0.00005,0.953799) (0.00006,0.946157) (0.00007,0.93876) (0.00008,0.931574) (0.00009,0.924572) (0.0001,0.917734) (0.00011,0.911043) (0.00012,0.904486) (0.00013,0.898052) (0.00014,0.891732) (0.00015,0.885517) (0.00016,0.879402) (0.00017,0.873378) (0.00018,0.867442) (0.00019,0.861589) (0.0002,0.855813) (0.00021,0.850112) (0.00022,0.844482) (0.00023,0.83892) (0.00024,0.833422) (0.00025,0.827987) (0.00026,0.822611) (0.00027,0.817293) (0.00028,0.812031) (0.00029,0.806822) (0.0003,0.801665) (0.00031,0.796557) (0.00032,0.791499) (0.00033,0.786487) (0.00034,0.781521) (0.00035,0.776599) (0.00036,0.771721) (0.00037,0.766884) (0.00038,0.762089) (0.00039,0.757333) (0.0004,0.752616) (0.00041,0.747937) (0.00042,0.743295) (0.00043,0.738689) (0.00044,0.734119) (0.00045,0.729583) (0.00046,0.725081) (0.00047,0.720612) (0.00048,0.716176) (0.00049,0.711772) (0.0005,0.707399) (0.00051,0.703057) (0.00052,0.698744) (0.00053,0.694462) (0.00054,0.690208) (0.00055,0.685983) (0.00056,0.681786) (0.00057,0.677616) (0.00058,0.673474) (0.00059,0.669358) (0.0006,0.665268) (0.00061,0.661204) (0.00062,0.657166) (0.00063,0.653153) (0.00064,0.649164) (0.00065,0.6452) (0.00066,0.641259) (0.00067,0.637342) (0.00068,0.633449) (0.00069,0.629578) (0.0007,0.625731) (0.00071,0.621905) (0.00072,0.618102) (0.00073,0.61432) (0.00074,0.61056) (0.00075,0.606821) (0.00076,0.603103) (0.00077,0.599406) (0.00078,0.59573) (0.00079,0.592073) (0.0008,0.588437) (0.00081,0.58482) (0.00082,0.581223) (0.00083,0.577646) (0.00084,0.574087) (0.00085,0.570548) (0.00086,0.567027) (0.00087,0.563524) (0.00088,0.56004) (0.00089,0.556574) (0.0009,0.553126) (0.00091,0.549696) (0.00092,0.546284) (0.00093,0.542888) (0.00094,0.539511) (0.00095,0.53615) (0.00096,0.532806) (0.00097,0.529479) (0.00098,0.526168) (0.00099,0.522874) (0.001,0.519597)								
	};
	\addlegendentry{$Q^{(1)}(\mathcal{B}_{p_X,100 p_X})$}

		
		\addplot[dotted,red,thick,smooth] coordinates {
(2.2204e-16,1)
(1e-05,0.98841)
(2e-05,0.97884)
(3e-05,0.97004)
(4e-05,0.96174)
(5e-05,0.9538)
(6e-05,0.94616)
(7e-05,0.93876)
(8e-05,0.93162)
(9e-05,0.92485)
(0.0001,0.91808)
(0.00011,0.91145)
(0.00012,0.90497)
(0.00013,0.89861)
(0.00014,0.89237)
(0.00015,0.88625)
(0.00016,0.88023)
(0.00017,0.8743)
(0.00018,0.86847)
(0.00019,0.86273)
(0.0002,0.85706)
(0.00021,0.85148)
(0.00022,0.84598)
(0.00023,0.84055)
(0.00024,0.83518)
(0.00025,0.82989)
(0.00026,0.82466)
(0.00027,0.81949)
(0.00028,0.81438)
(0.00029,0.80933)
(0.0003,0.80434)
(0.00031,0.7994)
(0.00032,0.79451)
(0.00033,0.78968)
(0.00034,0.7849)
(0.00035,0.78017)
(0.00036,0.77548)
(0.00037,0.77084)
(0.00038,0.76625)
(0.00039,0.7617)
(0.0004,0.7572)
(0.00041,0.75274)
(0.00042,0.74832)
(0.00043,0.74394)
(0.00044,0.7396)
(0.00045,0.7353)
(0.00046,0.73104)
(0.00047,0.72682)
(0.00048,0.72264)
(0.00049,0.71849)
(0.0005,0.71438)
(0.00051,0.7103)
(0.00052,0.70626)
(0.00053,0.70225)
(0.00054,0.69828)
(0.00055,0.69434)
(0.00056,0.69043)
(0.00057,0.68656)
(0.00058,0.68271)
(0.00059,0.6789)
(0.0006,0.67512)
(0.00061,0.67137)
(0.00062,0.66765)
(0.00063,0.66397)
(0.00064,0.66031)
(0.00065,0.65668)
(0.00066,0.65307)
(0.00067,0.6495)
(0.00068,0.64596)
(0.00069,0.64244)
(0.0007,0.63895)
(0.00071,0.63549)
(0.00072,0.63205)
(0.00073,0.62865)
(0.00074,0.62526)
(0.00075,0.62191)
(0.00076,0.61858)
(0.00077,0.61528)
(0.00078,0.612)
(0.00079,0.60875)
(0.0008,0.60552)
(0.00081,0.60232)
(0.00082,0.59914)
(0.00083,0.59599)
(0.00084,0.59286)
(0.00085,0.58976)
(0.00086,0.58667)
(0.00087,0.58362)
(0.00088,0.58059)
(0.00089,0.57758)
(0.0009,0.57459)
(0.00091,0.57163)
(0.00092,0.56869)
(0.00093,0.56577)
(0.00094,0.56287)
(0.00095,0.56)
(0.00096,0.55715)
(0.00097,0.55432)
(0.00098,0.55152)
(0.00099,0.54874)
(0.001,0.54597)
	};
		\addlegendentry{UB Theorem~\ref{thm:epsilonDegradableAdditivity}}			
		

	\end{axis}  

\end{tikzpicture}}\\
    \caption{This plot compares upper and lower bounds on $Q(\mathcal{B}_{p_X,p_Z})$ derived in~\eqref{eq_newUB_BB84} and in Theorem~\ref{thm:epsilonDegradableAdditivity}, for $\varepsilon_{p_X,p_Z}$ as in \eqref{opt:epsi}, with the best previously known upper bounds. We consider two different setups of $p_X/p_Z$.}
    \label{fig:bb84}
\end{figure}
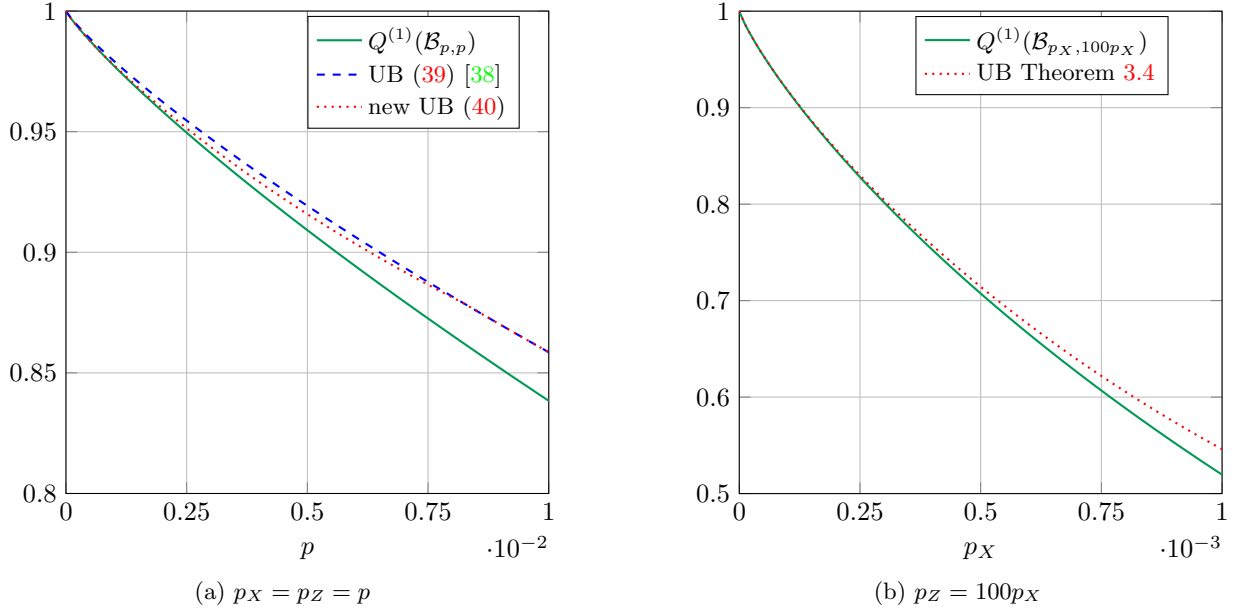

\subsection{Randomizing channels} \label{sec:randChannels}
A quantum channel $\Phi:\mathcal{S}(A) \to \mathcal{S}(B)$ is called \emph{$\varepsilon$-randomizing} if for any state $\rho \in \mathcal{D}(A)$, ${\|\Phi(\rho)-\tfrac{1}{|B|}\mathds{1}_{B}\|_{\opnorm}} \leq \tfrac{\varepsilon}{|B|}$, where $\tfrac{1}{|B|}\mathds{1}_{B}$ is the maximally mixed state on the system $B$.
Consider a channel $\Phi:\mathcal{S}(A) \to \mathcal{S}(B)$ whose complementary channel is $\Phi^{\setC}:\mathcal{S}(A) \ni  \rho \mapsto \tfrac{1}{|B|} \sum_{i=1}^{|B|} U_i \, \rho \, U_i^\dagger \in \mathcal{S}(E)$, where $\{U_i\}_{i=1}^{|B|}$ are independent random matrices Haar-distributed on the unitary group $\mathcal{U}(|A|)$ with $|A| = |E|$. Consider the fully mixing channel $\Xi:\mathcal{S}(B) \ni \rho \mapsto \Tr{\rho} \frac{1}{|E|} \mathds{1}_{E} \in \mathcal{S}(E)$.
\begin{myprop}[\cite{aubrun09}] \label{prop:aubrun}
If $|B| \geq C\tfrac{|E|^3}{\varepsilon^2}$ for some constant $C>0$ and $0 < \varepsilon <1$, then with high probability the channel $\Phi$ is $\varepsilon$-degradable.
\end{myprop}
\begin{proof}
Consider the two channels $\Phi$ and $\Xi$ as defined above. We show that if $|B| \geq C\tfrac{|E|^3}{\varepsilon^2}$, with high probability $\norm{\Phi^\setC - \Xi \circ \Phi}_{\diamond}\leq \varepsilon$. This follows directly from \cite[Thm.~1]{aubrun09} together with the fact that for two arbitrary cptp maps $\Theta_1, \Theta_2: \mathcal{S}(A) \to \mathcal{S}(B)$, 
\begin{align*}
\norm{\Theta_1 - \Theta_2}_{\diamond} 
&\leq |B| \max_{\rho \in \mathcal{D}(A)} \norm{\Theta_1(\rho)- \Theta_2(\rho)}_{\trnorm}\\
&\leq |B|^2  \max_{\rho \in \mathcal{D}(A)} \norm{\Theta_1(\rho)- \Theta_2(\rho)}_{\opnorm} \ .
\end{align*}
\end{proof}
We thus can use Theorem~\ref{thm:epsilonDegradableAdditivity} to estimate $Q(\Phi)$ and $P(\Phi)$ from above for most of the random unitary channels $\Phi$ as defined above that have an environment that is considerably smaller than the output system. As shown in \cite{aubrun09}, it can be verified that the constant $C$ in Proposition~\ref{prop:aubrun} can be chosen, e.g., as $C=150$.

\section{Discussion} \label{sec:conclusion}
We have seen that the concept of degradable channels can be generalized to the more robust notion of approximate degradable channels such that the beneficial additivity properties degradable channels offer are approximately preserved. As it can be efficiently determined how close in the diamond norm an arbitrary channel satisfies the degradability condition (by solving an SDP), the framework of approximate degradable channels can be used to derive upper bounds to the quantum and private classical capacity that can be evaluated efficiently.
Unlike previous attempts to derive upper bounds, our method does not rely on channel specific arguments and therefore can be applied to all channels.

For future work it would be of interest to better understand the differences between $\varepsilon$-degradable channels and $\varepsilon$-close-degradable channels (as introduced in Appendix~\ref{app:alternativeDef}). A problem that is left open is the question if a converse statement to Proposition~\ref{prop:relation} is possible, i.e., if an $\varepsilon$-degradable channel must be also $\theta(\varepsilon)$-close-degradable for some function $\theta: \Rp \ni \varepsilon \mapsto \theta(\varepsilon) \in \Rp$. Another question that deserves further investigation is if the optimization problem \eqref{opt:epsi2} can be solved (or at least approximated) efficiently. 

The concept of approximate degradable channels could also be useful in classical information theory. We note that to some extent, the current understanding about the capacity region of a classical broadcast channel is comparable to the knowledge about the quantum and private classical capacity --- no single-letter formula is known, except in the case of a degradable broadcast channel \cite{elgamal12}. As such it could be promising to apply the framework of approximate degradable quantum channels, introduced in this article, to classical broadcast channels. 

Very recently, the bounds for approximate degradable channels derived in this article have been applied to determine the quantum and the private capacity of low-noise quantum channels to leading orders in the channelÕs distance to the perfect channel~\cite{letditzky17_2}.

\section*{Acknowledgments}
We would like to thank Omar Fawzi, Philipp Kammerlander, Graeme Smith, Marco Tomamichel, and Michael Wolf for helpful discussions and pointers to references. We further thank the associate editor Mark M.~Wilde for constructive feedback.

This project was supported by the Swiss National Science Foundation (through the National Centre of Competence in Research `Quantum Science and Technology'), by the European Research Council (grant No.~258932), and by the Air Force Office of Scientific Research (AFOSR) via grant FA9550-16-1-0245. VBS acknowledges support by an ETH postdoctoral fellowship. AW's work was supported by the EU (STREP ``RAQUEL''), the ERC (AdG ``IRQUAT''), the Spanish MINECO (grant FIS2013-40627-P) with the support of FEDER funds, as well as by the Generalitat de Catalunya CIRIT, project~2014-SGR-966.



\appendix

\section{Approximate degradabiliy versus closeness to degradable channels} \label{app:alternativeDef}

The operational meaning behind Definition~\ref{def:epsiDeg} is that a channel is called $\varepsilon$-degradable if the degradability condition is approximately (up to an $\varepsilon$) satisfied. An alternative approach is to consider the distance to degradable channels.

\begin{mydef} [$\varepsilon$-close-degradable] \label{def:EpsiCloseDeg}
A channel $\Phi:\mathcal{S}(A) \to \mathcal{S}(B)$ is said to be \emph{$\varepsilon$-close-degradable} if there exists a degradable channel $\Psi:\mathcal{S}(A) \to \mathcal{S}(B)$ such that $\norm{\Phi-\Psi}_{\diamond} \leq \varepsilon$.
\end{mydef}

By Definition~\ref{def:EpsiCloseDeg} and since the identity channel is degradable, it follows that every channel $\Phi:\mathcal{S}(A) \to \mathcal{S}(B)$ is $\varepsilon$-close-degradable with respect to some $\varepsilon \in [0,2]$.
The following proposition proves that $\varepsilon$-close-degradable channels, similar as $\varepsilon$-degradable channels, inherit the additivity properties of degradable channels with an error term that vanishes in the limit $\varepsilon \to 0$.
\begin{myprop}[Properties of $\varepsilon$-close-degradable channels] \label{prop:epsiDeg2}
Let $\Phi:\mathcal{S}(A) \to \mathcal{S}(B)$ be a quantum channel that is $\varepsilon$-close-degradable (with respect to a degradable channel $\Psi:\mathcal{S}(A) \to \mathcal{S}(B)$), then
\begin{enumerate}[(i)]
\item $|Q(\Phi) - Q^{(1)}(\Psi)|  \leq  \varepsilon \log |B| + (2+\varepsilon) \Hb\bigl(\frac{\varepsilon}{2+\varepsilon}\bigr)$ , \label{it:si}
\item $|P(\Phi)- Q^{(1)}(\Psi) | \leq 2\varepsilon \log |B| + 2(2+\varepsilon) \Hb\bigl(\frac{\varepsilon}{2+\varepsilon}\bigr)$\, . \label{it:siii}
\end{enumerate}
\end{myprop}
\begin{proof}
We first prove statement \eqref{it:si} of the proposition. Since the coherent information is additive for degradable channels,
\begin{align*}
|Q(\Phi)-Q^{(1)}(\Psi)| 
&= \left| Q(\Phi)- Q(\Psi) \right| \\
&\leq \varepsilon \log |B|  + (2+\varepsilon) \Hb\Bigl(\frac{\varepsilon}{2+\varepsilon}\Bigr) \, , 
\end{align*}
where the inequality is due to \cite[Cor.~2]{leung09} together with Lemma~\ref{lem:alickiFannes}. 



Statement \eqref{it:siii} of the proposition can be proven as follows. As $\Psi$ is degradable we have $P(\Psi)=Q^{(1)}(\Psi)$ which gives 
\begin{align*}
|P(\Phi)-Q^{(1)}(\Psi)| 
&= |P(\Phi)-P(\Psi)|  \\
&\leq 2\varepsilon \log |B| + 2(2+\varepsilon) \Hb\Bigl(\frac{\varepsilon}{2+\varepsilon}\Bigr)\, ,
\end{align*}
where the inequality follows from the proof of \cite[Cor.~3]{leung09} with Lemma~\ref{lem:alickiFannes}.
%
%
%
%
\end{proof}

As mentioned in Section~\ref{sec:preliminaries}, the function $\rho \mapsto \Icc{\rho}{\Psi}$ is concave if $\Psi$ is degradable which can be helpful when computing the channel coherent information given in \eqref{eq:channelCoherentInfo} (see Proposition~\ref{prop:epsiDeg2}).
We emphasize that the proofs of Theorem~\ref{thm:epsilonDegradableAdditivity} and Proposition~\ref{prop:epsiDeg2} are different although they both prove a similar statement, however under different assumptions. The proof of Theorem~\ref{thm:epsilonDegradableAdditivity} generalizes Devetak and Shor's proof for additivity of degradable channels \cite{shor2short}, whereas the proof of Proposition~\ref{prop:epsiDeg2} is based on continuity properties of channel capacities, following~\cite{leung09} and using the improved Alicki-Fannes inequality (Lemma~\ref{lem:alickiFannes}).

%

Definition~\ref{def:EpsiCloseDeg} directly implies that if a channel $\Phi:\mathcal{S}(A) \to \mathcal{S}(B)$ is $\varepsilon$-close-degradable it is also  $\varepsilon'$-close-degradable for all $\varepsilon' \geq \varepsilon$. The smallest possible value $\varepsilon$ such that $\Phi$ is $\varepsilon$-close-degradable is given by
\begin{align} \label{opt:epsi2}
\hat \varepsilon_{\Phi}:=\left \lbrace
\begin{array}{lll}
			&\inf\limits_{\Psi,\Theta} 		& \norm{\Phi-\Psi}_{\diamond} \\
			&\st					& \Psi^{\setC} = \Theta \circ \Psi \\
			& & \Psi:\mathcal{S}(A) \to \mathcal{S}(B) \textnormal{ is cptp}\\
			& & \Theta:\mathcal{S}(B) \to \mathcal{S}(E) \textnormal{ is cptp}\, .
	\end{array} \right.
\end{align}
Note that unlike the optimization problem \eqref{opt:epsi} which can be phrased as an SDP and as a consequence can be solved efficiently, it is unclear if $\hat \varepsilon_{\Phi}$ can be computed efficiently. The optimization problem \eqref{opt:epsi2} is clearly not an SDP as the constraint $\Psi^\setC = \Theta\circ \Psi$ is not linear in $(\Psi,\Theta)$. 
%

Similarly to an $\varepsilon$-close-degradable channel being defined via being close to a degradable channel we can define an $\varepsilon$-close-anti degradable channel as being close to an anti-degradable channel.
\begin{mydef} [$\varepsilon$-close-anti-degradable] \label{def:EpsiCloseAntiDeg}
A channel $\Phi:\mathcal{S}(A) \to \mathcal{S}(B)$ is said to be \emph{$\varepsilon$-close-anti-degradable} if there exists an anti-degradable channel $\Xi:\mathcal{S}(A) \to \mathcal{S}(B)$ such that $\norm{\Phi-\Xi}_{\diamond} \leq \varepsilon$.
\end{mydef}

\begin{mycor}[Properties of $\varepsilon$-close-anti-degradable channels] \label{cor:closeAntiDeg}
Let $\Phi:\mathcal{S}(A) \to \mathcal{S}(B)$ be a quantum channel that is $\varepsilon$-close-anti-degradable, then $Q(\Phi) \leq P(\Phi) \leq 2\varepsilon \log |B| + 2(2+\varepsilon) \Hb\bigl(\frac{\varepsilon}{2+\varepsilon}\bigr)$.
\end{mycor}
\begin{proof}
This corollary follows immediately from \cite[Cor.~3]{leung09} together with Lemma~\ref{lem:alickiFannes} and the fact that anti-degradable are known to have a private capacity that is zero \cite{smith08_3}. 
\end{proof}
Similar as above, given a channel $\Phi:\mathcal{S}(A) \to \mathcal{S}(B)$ the smallest parameter $\varepsilon$ such that $\Phi$ is $\varepsilon$-close-anti-degradable is given by
\begin{align*}
\tilde\varepsilon_{\Phi}:=\left \lbrace
\begin{array}{lll}
			&\inf\limits_{\Xi} 		& \norm{\Phi-\Xi}_{\diamond} \\
			&\st					& \Xi = \Theta \circ \Xi^{\setC} \\
			& & \Xi:\mathcal{S}(A) \to \mathcal{S}(B) \textnormal{ is cptp}\\
			& & \Theta:\mathcal{S}(E) \to \mathcal{S}(B) \textnormal{ is cptp}\, .
	\end{array} \right.
\end{align*}
The close connection between anti-degradability and $2$-extendibility may be helpful to (efficiently) compute the quanitity $\tilde\varepsilon_{\Phi}$ (see~\cite[Lemma B.1]{letditzky17}).\footnote{We would like to thank one referee for pointing this out.}
Corollary~\ref{cor:closeAntiDeg} implies that for an arbitrary channel $\Phi:\mathcal{S}(A) \to \mathcal{S}(B)$ with $|B|:=\dim B$, we have $Q(\Phi) \leq 2 \tilde \varepsilon_{\Phi} \log |B| + 2(2+\tilde\varepsilon_{\Phi}) \Hb\bigl(\frac{\tilde\varepsilon_{\Phi}}{2+\tilde\varepsilon_{\Phi}}\bigr)$.
There is a close connection between the two concepts of an $\varepsilon$-degradable and an $\varepsilon$-close-degradable channel as stated in the following proposition.
\begin{myprop}[Relation between $\varepsilon$-close-degradable and $\varepsilon$-degradable] \label{prop:relation}
Let $\Phi:\mathcal{S}(A) \to \mathcal{S}(B)$ be a quantum channel that is $\varepsilon$-close-degradable, then $\Phi$ is $(\varepsilon+2\sqrt{\varepsilon})$-degradable.
\end{myprop}
\begin{proof}
Let $\Phi_1,\Phi_2:\mathcal{S}(A) \to \mathcal{S}(B)$ be two channels such that $\norm{\Phi_1 - \Phi_2}_{\diamond} \leq \varepsilon$. Then, by the continuity of Stinespring's representation \cite[Equation~(2)]{werner08} it follows that there exist two complementary channels $\Phi_1^\setC$ and $\Phi_2^\setC$ such that $\norm{\Phi_1^\setC - \Phi_2^\setC}_{\diamond}\leq 2 \sqrt{\varepsilon}$. By assumption there exist two channels $\Xi:\mathcal{S}(A) \to \mathcal{S}(B)$ and $\Theta:\mathcal{S}(B) \to \mathcal{S}(E)$ such that $\norm{\Phi - \Xi}_{\diamond}\leq \varepsilon$ with $\Xi^\setC = \Theta \circ \Xi$. As explained above, the continuity of Stinespring's representation implies that 
\begin{equation*}
\norm{\Phi^\setC - \Theta \circ \Xi}_{\diamond} = \norm{\Phi^\setC - \Xi^\setC}_{\diamond} \leq 2 \sqrt{\varepsilon}\, .
\end{equation*}
Using the triangle inequality gives
\begin{align*}
\norm{\Phi^\setC - \Theta \circ \Phi}_{\diamond} &\leq \norm{\Phi^\setC - \Theta \circ \Xi}_{\diamond} + \norm{\Theta \circ \Xi - \Theta \circ \Phi}_{\diamond}\\
&\leq 2 \sqrt{\varepsilon} + \norm{ \Xi - \Phi}_{\diamond}\\
&\leq 2 \sqrt{\varepsilon} + \varepsilon\, ,
\end{align*}
which proves the assertion.
\end{proof}
It is unclear whether a converse statement to the one given in Proposition~\ref{prop:relation}, i.e., that a channel being $\varepsilon$-degradable implies that it is $\theta(\varepsilon)$-close-degradable for some function $\theta:[0,1]\to \Rp$, is valid. This seems to be difficult to prove (if possible at all) as the structure of the set of degradable channels is poorly understood. (It is known that the set of degradable channels is not convex \cite{cubitt08}, however it is unknown how large this set is.)


\section{Alternative proof of claim~\eqref{it:i} of Theorem~\ref{thm:epsilonDegradableAdditivity}} \label{app:alternativeProof}
For the sake of completeness we present here an alternative proof of claim~\eqref{it:i}.
The lower bound $Q^{(1)}(\Phi)\leq Q(\Phi)$ is immediate. 
Consider a pure state $\phi_{A A_1' \ldots A'_n}$ and for $i \in [n]$ let $V^i_{A'_i \to B_i E_i}$ denote the isometric extension of the $i$-th channel $\Phi$. For $i\in [n]$ let $\sigma^i := V^i \phi (V^i)^{\dagger}$, $\rho_{A B_1 E_1 \ldots B_n E_n}:=(\bigotimes_{i=1}^n V^i )\phi (\bigotimes_{i=1}^n {V^i}^{\dagger})$ and define 
\begin{equation*}
\xi:= \frac{\varepsilon}{2}  \log(|E|-1) +  \Hb(\varepsilon/2)  +\varepsilon \log |E| + \Bigl(1+\frac{\varepsilon}{2} \Bigr) \Hb\Bigl(\frac{\varepsilon}{2+\varepsilon} \Bigr) \, .
\end{equation*}
Assuming that $\rho_{A B_1 E_1 \ldots B_n E_n}$ is the state that maximizes $Q^{(1)}(\Phi^{\otimes n})$ gives
\begin{align}
Q^{(1)}(\Phi^{\otimes n})
&= I(A \rangle B_1 \ldots B_n)_{\rho}\nonumber \\
 &= \Hh{B_1 \ldots B_n}_{\rho} - \Hh{A B_1 \ldots B_n}_{\rho} \nonumber\\
 &=\Hh{B_1 \ldots B_n}_{\rho} - \Hh{E_1 \ldots E_n}_{\rho} \label{eq:pure}\\
 &=\sum_{i=1}^n \Hh{B_i}_{\rho} - \Hh{E_i}_{\rho}  - \left(I(B_1: B_2 : \ldots : B_n)_{\rho}-I(E_1: E_2 : \ldots : E_n)_{\rho}  \right) \nonumber\\
 & \leq \sum_{i=1}^n \Hh{B_i}_{\rho} - \Hh{E_i}_{\rho} + n \xi \label{ineq:lemma}\\
 &= \sum_{i=1}^n \Hh{B_i}_{\sigma^i} - \Hh{A B_i A'_1 \ldots A'_{i-1} A'_{i+1} \ldots A'_n}_{\sigma^i} + n \xi \label{eq:pure2}\\
 &= \sum_{i=1}^n I(A A'_1 \ldots A'_{i-1} A'_{i+1} \ldots A'_n \rangle B_i)_{\sigma^i} +n \xi \nonumber\\
 &\leq n Q^{(1)}(\Phi) + n \xi \, ,\nonumber
\end{align}
where \eqref{eq:pure} follows since the state $\rho$ is pure on the system $A B_1 E_1 \ldots B_n E_n$. Inequality \eqref{ineq:lemma} follows by $n$ times applying Lemma~\ref{lem:help}. Equation \eqref{eq:pure2} is true since the entropies of $\rho$ and $\{ \sigma^i\}_{i=1}^n$ on the given reduced systems are equal and $\sigma^i$ is pure on $A A'_1\ldots A'_{i-1} B_i E_i A'_{i+1} \ldots A'_n$. The final inequality follows as the states $\{\sigma^i\}_{i=1}^n$ are not necessarily the optimizers for the corresponding coherent informations. This proves statement \eqref{it:i} of Theorem~\ref{thm:epsilonDegradableAdditivity}.

\bibliographystyle{abbrv}
\bibliography{./bibtex/header,./bibtex/bibliofile}

\end{document}